\documentclass[a4paper,UKenglish]{lipics-v2018}

\usepackage{microtype}
\usepackage{verbatim}
\usepackage{setspace}
\usepackage{etoolbox}
\usepackage{pwbib}
\usepackage{hyperref}
\usepackage{pwmath}
\usepackage[vlined,noresetcount]{pwalgorithm}
\usepackage{paralist}
\usepackage[noabbrev,
capitalise,
]{cleveref} 
\crefname{line}{line}{lines}
\Crefname{line}{Line}{Lines}
\crefname{table}{Table}{Table}

\crefname{page}{page}{pages}
\Crefname{page}{Page}{Pages}
\crefname{equation}{}{}
\Crefname{equation}{}{}
\crefname{observation}{Observation}{Observations}
\Crefname{observation}{Observation}{Observations}
\crefname{claim}{Claim}{Claims}
\Crefname{claim}{Claim}{Claims}
\crefname{lemma}{Lemma}{Lemmas}
\Crefname{lemma}{Lemma}{Lemmas}

\newtheorem{claim}[theorem]{Claim}
\newtheorem{observation}[theorem]{Observation}

\newcommand{\RMR}{\mathit{RMR}}
\newcommand{\RR}{\mathcal{R}}
\newcommand{\Exec}{Exec}
\newcommand{\Conf}{Conf}
\newcommand{\Sched}[1]{\sigma_{\rightarrow #1}}
\newcommand{\E}[1]{E_{\rightarrow #1}}
\newcommand{\Val}{val}
\newcommand{\Lost}{L}
\newcommand{\K}{K}
\newcommand{\HH}{H}
\newcommand{\Proc}{Proc}
\newcommand{\Cache}{Cache}
\newcommand{\VV}{\mathcal{V}}

\SetKwFunction{NameDecide}{NameDecide}

\title{An Almost Tight RMR Lower Bound for Abortable Test-And-Set}

\titlerunning{RMR Lower Bound for Abortable TAS}

\author{Aryaz Eghbali}{Department of Computer Science, University of Calgary, Canada}{aryaz.eghbali@ucalgary.ca}{}{}
\author{Philipp Woelfel}{Department of Computer Science, University of Calgary, Canada}{woelfel@ucalgary.ca}{}{}

\authorrunning{A. Eghbali and P. Woelfel}

\Copyright{Aryaz Eghbali and Philipp Woelfel}

\subjclass{}


\category{}

\relatedversion{}

\supplement{}

\funding{This research was undertaken, in part, thanks to funding from the Canada Research
Chairs program and from the Discovery Grants program of the Natural Sciences and
Engineering Research Council of Canada (NSERC).}

\ArticleNo{0}
\nolinenumbers 
\hideLIPIcs  

\begin{document}

\maketitle

\begin{abstract}
  We prove a lower bound of $\Omega(\log n/\log\log n)$ for the remote memory reference (RMR) complexity of abortable test-and-set (leader election) in the cache-coherent (CC) and the distributed shared memory (DSM) model.
  This separates the complexities of abortable and non-abortable test-and-set, as the latter has constant RMR complexity \cite{GHW2010a}.
  
  Golab, Hendler, Hadzilacos and Woelfel \cite{GHHW2012a} showed that compare-and-swap can be implemented from registers and TAS objects with constant RMR complexity.
  We observe that a small modification to that implementation is abortable, provided that the used TAS objects are abortable. 
\end{abstract}

\section{Introduction}
In this paper, we study the remote memory references (RMR) complexity of abortable test-and-set.
Test-and-set (TAS) is a fundamental shared memory primitive that has been widely used as a building block for classical problems such as mutual exclusion and renaming, and for the construction of stronger synchronization primitives \cite{KRS1988a,PPTV1998a,EHW1998a,BPSV2006a,AAGGG2010a,AAGG2011a,AACGZ2011a,GHHW2012a}.

We consider a standard asynchronous shared memory system in which $n$ processes with unique IDs communicate by reading and writing shared registers.
A TAS object stores a bit that is initially 0, and provides two methods, \TAS{}, which sets the bit and returns its previous value, and \Read{}, which returns the current value of the bit.
TAS is closely related to mutual exclusion \cite{Dij1965a}: a TAS object can be viewed as a one-time mutual exclusion algorithm, where only one process (the one whose \TAS{} returned 0) can enter the critical section \cite{DHW1997a}.

TAS objects have consensus-number two, and therefore they have no wait-free implementations.
In particular, in deterministic TAS implementations, processes may have to wait indefinitely, by spinning (repeatedly reading) variables.
It is common to predict the performance of such blocking algorithms by bounding remote memory references (RMRs).
These are memory accesses that traverse the processor-to-memory interconnect.
Local-spin algorithms achieve low RMR complexity by spinning on locally accessible variables.
Two models are common:
In \emph{distributed shared memory (DSM)} systems, each shared variable is permanently locally accessible to a single processor and remote to all other processors.
In \emph{cache-coherent (CC)} systems, each processor keeps local copies of shared variables in its cache; the consistency of copies in different caches is maintained by a \emph{coherence protocol}.
Memory accesses that cannot be resolved locally and have to traverse the processor-to-memory interconnect are called \emph{remote memory references} (RMRs).

Golab, Hendler, and Woelfel \cite{GHW2010a} devised deadlock-free TAS algorithms with $O(1)$ RMR complexity for the DSM and the CC model, which in turn have been used to construct equally efficient comparison-primitives, such as compare-and-swap (CAS) objects \cite{GHHW2012a}.
These constructions are particularly useful in the study of the complexity of the mutual exclusion problem, for which the RMR complexity is the standard performance metric  \cite{AK2002a,AK2000a,KA2006a,AHW2008a,DG2010a,Jay2003a,JPN2005a,KA2001a,BG2011a,HW2010a,HW2011a,PW2012a,GW2012a,And1990a,Lee2005a,DL2008a,Lee2010a,GW2014a}.

In the context of mutual exclusion, it has been observed that systems often require locks to support a ``timeout'' capability that allows a process waiting too long for the lock, to abort its attempt \cite{Sco2002a}.
In database systems, such as Oracle's Parallel Server and IBM's DB2, the ability of a thread to abort lock attempts serves the dual purpose of recovering from a transaction deadlock and tolerating preemption of the thread that holds the lock \cite{Sco2002a}.
In real time systems, the abort capability can be used to avoid overshooting a deadline.
Solutions to this problem have been proposed in the form of \emph{abortable} mutual exclusion algorithms \cite{Sco2002a,Jay2003a,PW2012a,Lee2010a,DL2008a,GW2017a}.
In such an algorithm, at any point a process may receive an \emph{abort signal} upon which, within a finite number of its own steps, it must either enter the critical section or abort its current attempt to do so, by returning to the remainder section.

The complexity of the mutual exclusion problem is not affected by abortability: The abortable algorithm by Danek and Lee \cite{DL2008a,Lee2011a} achieves $O(\log n)$ RMR complexity, which asymptotically matches the known lower bound for non-abortable mutual exclusion \cite{AHW2008a}.
But abortable mutual exclusion algorithms seem to be much more difficult to obtain than non-abortable ones, and it is not surprising that all such algorithms preceding \cite{DL2008a,Lee2011a} used stronger synchronization primitives (e.g., LL/SC objects in \cite{Jay2003a}).
Moreover, no RMR efficient randomized abortable mutual exclusion algorithms are known, unless stronger primitives are used \cite{PW2012a,GW2017a};
on the other hand, several non-abortable randomized implementations use only registers~\cite{HW2009a,HW2010a,GW2014a,BG2011a}.

As mentioned earlier, CAS objects with $O(1)$ RMR complexity can be obtained from registers \cite{GHHW2012a}, but they cannot be used in an abortable mutual exclusion algorithm without sacrificing its abortability: if a process receives the abort signal while being blocked in an operation on a CAS object, it has no option to finish that operation in a wait-free manner, and thus can also not abort its attempt to enter the critical section.
In general, implemented blocking strong objects, cannot be used to obtain abortable mutual exclusion objects.

One way of dealing with this impasse can be to make implementations of strong primitives also abortable, and to devise mutual exclusion algorithms in such a way that they accommodate operation aborts.
Similarly, other algorithms and data structures that may require timeout capabilities, can potentially be implemented from abortable objects, but not from non-abortable ones.

We define abortability in the following, natural way:
In a concurrent execution, a process executing an operation on the object may receive an abort signal at any point in time.
When that happens, it must finish its method call within a finite number of its own steps (wait-free), and as a result the method call may fail to take effect, or it may succeed.
The resulting execution must satisfy the safety conditions of the object (e.g., linearizability), if all failed operations are removed.
Moreover, a process must be able to find out, by looking at the return value, whether its aborted operation succeeded, and if it did, then the return value must be consistent with a successful operation.

It may be tempting to define a weaker forms of abortability, e.g., where a return value of an aborted operation does not indicate whether the operation succeeded or not.
But the usefulness of such a weaker notions is not clear. 
For example, abortable TAS objects (according to our definition) can easily be used to implement an abortable mutual exclusion algorithm (TAS-lock): One can store a pointer to a ``current'' TAS object in a single register $R$. 
To get the lock, a process calls \TAS{} on the TAS object that $R$ points to, and if the return value  is 0, then the process has the lock, and otherwise it keeps reading $R$ until its value changes.
To release the lock, the process simply swings the pointer $R$ so that it points to a new, fresh TAS object
(this technique was proposed in \cite{AA2011a}, and \cite{AGW2014a,AW2014a} showed how to bound the number of involved TAS objects).
This also works in the case of aborts, because a process knows whether its operation took effect, and thus whether it is allowed to swing the pointer (and in fact must, to avoid dead-locks).

For the weaker definition of abortability mentioned above, a process whose \TAS{} aborted may not be able to find out whether it has the lock or not, and then it can also not swing the pointer to a new TAS object, even though its \TAS{} may have set the bit from 0 to 1.
In fact, suppose that two processes call \TAS{}, and both \TAS{} calls abort without receiving the information whether the aborted operation took effect.
Then the TAS bit may be set, but none of the processes has received any information regarding who was successful, and reading the TAS object also provides no information.

Even though our notion of abortability may seem strong, any abortable mutual exclusion algorithm can be used to obtain any abortable object from its corresponding sequential implementation, by simply protecting the sequential code in the critical section.
An interesting question is therefore, whether abortable objects can be obtained at a lower RMR cost than mutual exclusion.

We observe that this is true for implementations of abortable CAS objects from abortable TAS objects on the CC model: a straight-forward modification of the constant RMR implementation of non-abortable CAS from TAS objects and registers \cite{GHHW2012a}, immediately yields an abortable CAS object, provided that the used TAS objects are atomic or also abortable.
\begin{theorem}\label{thm:CAS_upper_bound}
  There is a deadlock-free implementation of abortable CAS from atomic registers and deadlock-free abortable TAS objects, which has $O(1)$ RMR complexity on the CC model.
\end{theorem}

Note that there are efficient randomized implementations of TAS from registers, where the maximum number of steps any process takes in a \TAS{} operation is $O(\log^\ast n)$ against an oblivious adversary \cite{GW2012b}.
In the construction of CAS above, we can use such a randomized TAS implementation in place of abortable TAS.
\begin{corollary}
  There is a deadlock-free randomized implementation of abortable CAS from atomic registers, such that on the CC model against an oblivious adversary each abort is randomized wait-free, and each operation on the object incurs at most $O(\log^\ast n)$ RMRs.
\end{corollary}

Recall that there is also a deterministic constant RMR implementation of TAS from registers \cite{GHW2010a}, and thus making this implementation abortable, would, together with the result mentioned above, immediately yield deterministic constant RMR abortable implementations of CAS from registers.
Unfortunately, it turns out that a deterministic constant RMR implementation of abortable TAS from registers cannot exist.
In particular, we define the abortable leader election (LE) problem, which is not harder and possibly easier than abortable TAS (with respect to RMR complexity).
Our main technical result is an RMR lower bound of $\Omega(\log n/\log\log n)$ for that object.

In a (non-abortable) LE protocol, every process decides for itself whether it becomes the leader (it returns $win$) or whether it loses (it returns $lose$).
At most one process can become the leader, and not all participating processes can lose.
I.e., if all participating processes finish the protocol, then exactly one of them returns $win$ and all others return $lose$.
Note that then in an abortable LE protocol all participating processes allowed to return $lose$, provided that all of them received the abort signal.


An abortable TAS object immediately yields an abortable LE protocol:
Each process executes a single \TAS{} operation and returns $win$ if the \TAS{} call returns 0, and otherwise $lose$ (i.e., it returns $lose$ also when the \TAS{} return value indicates a failed abort).

Our main result is the following:
\begin{theorem}\label{thm:main_lower_bound}
  For both, the DSM and the CC model, any deadlock-free abortable leader election algorithm has an execution in which at least one process incurs $\Omega(\log n/\log\log n)$ RMRs.
\end{theorem}

Leader election is one of the seemingly simplest synchronization primitives that have no wait-free implementation.
In particular, as argued above, the lower bound in \cref{thm:main_lower_bound} immediately also applies to abortable TAS. 
This is in stark contrast to the $O(1)$ RMRs upper bound for non-abortable TAS and even CAS implementations \cite{GHW2010a,GHHW2012a}.
It shows that adding abortability to synchronization primitives is almost as difficult as solving abortable mutual exclusion, which has an RMR complexity of $\Theta(\log n)$ \cite{DL2008a,Lee2011a}.

In our lower bound proof we identify the crucial reason for why abortable LE is harder than its non-abortable variant:
According to standard bi-valency arguments, for any deadlock-free LE algorithm, there is an execution in which some process takes an infinite number of steps.
But it is not hard to see that one can design an (asymmetric) 2-process LE protocol in which one fixed process is wait-free, because the other one waits for the first one to make a decision if it detects contention.
It turns out that this is not the case for abortable LE: Here, for any process, there is an execution in which that process takes an infinite number of steps.
%

\paragraph*{Other Related Work.}
Aguilera, Fr{\o}lund, Hadzilacos, Horn, and Toueg \cite{AFHHT2007a} define a different notion of abortable object, where no abort signals are sent by the system, but a process may decide for itself to abort an ongoing operation, e.g., when it detects contention.
According to their definition, the caller of an aborted operation may not find out whether its operation took effect or not.
Since this uncertainty may not be acceptable, they also introduce query-abortable objects, where a query operation allows a process to determine additional information about its last non-query operation.

Note that their notion of abortability is quite different from the one used commonly for mutual exclusion and adopted by us, where the system, and not the implementation, dictates when a process needs to abort.

\section{Abortable Compare-And-Swap in the CC Model}
In this section we consider the cache-coherent (CC) model.
Each process obtains a cache-copy with each read of a register, and the cache-copy gets only invalidated if some process later writes to the same register.
Writes as well as reads of non-cached registers incur RMRs, while reads of cached registers do not.

A CAS object provides two operations, \CAS{$cmp,new$}, and \Read{}.
Operation \Read{} returns the current value of the object.
Operation \CAS{$cmp,new$} writes $new$ to the object, if the current value is $cmp$, and otherwise does not change the value of the object.
In either case it returns the old value of the object.

Golab et.\,al.~\cite{GHHW2007a} gave an implementation of CAS from TAS and registers, which has constant RMR complexity in the CC model, i.e., each \CAS{} and reach \Read{} operation incurs only $O(1)$ RMRs.
In this section we show how to make that implementation abortable, provided that we have access to abortable TAS objects.
The pseudocode is in \cref{fig:algorithms}.
The original (non-abortable) version of the code is shown in black and our additional code to make it abortable in \textcolor{red}{red} (lines~\ref{line:NameDecide:abort} and \ref{line:CAS:abort}).

\begin{figure}[h!]
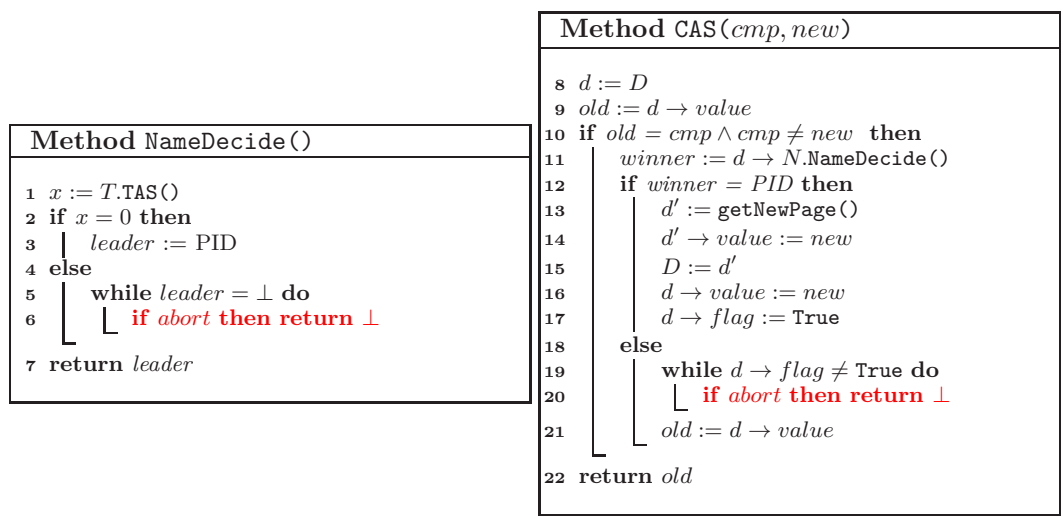
\footnotesize
  \NoCaptionOfAlgo\LinesNotNumbered\RestyleAlgo{boxruled}
  \begin{minipage}{.49\textwidth}
%
    \begin{method}{NameDecide()}
      $x := T.\FuncSty{TAS()}$\;
      \eIf{$x = 0$}
      {$leader$ := PID}
      {
        \While{$leader = \bot$}{
          \textcolor{red}{\IlIf{$abort$}{\Return $\bot$}}\label{line:NameDecide:abort}
        }
      }
      \Return{leader}
    \end{method}
  \end{minipage}  
  \begin{minipage}{.49\textwidth}
    \begin{method}{CAS($cmp,new$)}
      $d := D$\;
      $old := d\rightarrow value$\;
      \If{old = $cmp \wedge cmp \neq new$ \label{code:CAS:checkOld}}
      {
       $winner := d\rightarrow N.$\FuncSty{NameDecide()}\;
       \eIf{winner = PID}
        {
          $d' := \FuncSty{getNewPage()}$\;
          $d'\rightarrow value := new$\;
          $D := d'$\;
          $d\rightarrow value:=new$\;
          $d\rightarrow flag:=\True$\;
        }{
          \While{$d\rightarrow flag\neq\True$}{
            \textcolor{red}{\IlIf{$abort$}{\Return{$\bot$}}}\label{line:CAS:abort}
          }
          $old:=d\rightarrow value$\;
        }
      }
      \Return{$old$}
      \end{method}
\end{minipage}  

\caption{Implementation of \textcolor{red}{(abortable)} \FuncSty{NameDecide()} and \CAS{}.
  Without lines~\ref{line:NameDecide:abort} and \ref{line:CAS:abort} the algorithms are equivalent to the non-abortable implementations in \cite{GHHW2007a}.
  \label{fig:algorithms}
}
\end{figure}

\subsection{From TAS to Name Consensus}
The implementation in \cite{GHHW2007a} first constructs a name consensus object from a single TAS object $T$.
This implementation provides a method \NameDecide{}, which each process is allowed to call at most once.
All \NameDecide{} calls return the same value (agreement), which is the ID of a process calling \NameDecide{} (validity).

The non-abortable implementation in \cite{GHHW2007a} uses a TAS object $T$ and a register $leader$ that is initially $\bot$.
In a \NameDecide{} call, a process $p$ first calls $T$.\TAS{}.
If the \TAS{} returns 0, then $p$ \emph{wins}, and writes $p$ to $leader$.
Otherwise, $p$ \emph{loses}, and so it repeatedly reads $leader$, until $leader\neq \bot$, upon which $p$ can return the value of $leader$.
It is easy to see (and was formally proved in \cite{GHHW2007a}) that this is a correct name consensus algorithm.

We now show how this implementation can be made abortable, assuming the TAS object $T$ is abortable.
We assume that when a process receives the abort signal, a static process-local variable $abort$, which is initially false, changes to \True.

Recall that abortability requires that the return value of a \TAS{} operation indicates whether it failed or succeeded.
We assume a failed \TAS{} simply returns $\bot$.
In \NameDecide{}, processes are only waiting until $leader$ changes.
If a process is receiving the abort signal while waiting for $leader$ to change, then it can also simply return $\bot$.
The rest of the algorithm is the same as the original name consensus algorithm.

Clearly, the new code (\cref{line:NameDecide:abort}) does not affect RMR complexity, and following an abort the code is wait-free.
Moreover, correctness (validity and agreement) in case of no failed \NameDecide{} operations follow immediately from correctness of the original algorithm.
If a \NameDecide{} operation fails (i.e., returns $\bot$), then it did not change any shared memory object (its \TAS{} must have either failed, or returned 1).
Hence, removing an aborted and failed \NameDecide{} operation from the execution does not affect any other processes, and therefore the resulting execution must be correct.

\subsection{From Name Consensus to Compare-And-Swap}
We now show how the abortable name consensus algorithm can be used to obtain abortable CAS.
Consider the implementation of \CAS{$cmp,new$} on the right hand side in \cref{fig:algorithms}.
The black code is logically identical to the one in \cite{GHHW2007a}.
It uses a register $D$ that points to a \emph{page}, which stores two registers, $value$ and $flag$, as well as a name consensus object $N$.
Register $value$ at the page pointed to by $D$ stores the current value of the object.
(Thus, a \Read{} operation, for which we omit the pseudo code, simply returns  $D\rightarrow value$.)
The \CAS{} operation  assumes a wait-free method \FuncSty{getNewPage()}, which returns an unused page from a pool of pages (for simplicity assume this pool has infinitely many pages, but there are methods for wait-free memory management that allow using a bounded pool \cite{GHHW2012a,AW2016a}).

For a description of how the algorithm \CAS{$cmp,new$} works, we refer to \cite{GHHW2007a}.
We can prove that the abortable version presented here is correct, provided that the non-abortable version (with line~\ref{line:CAS:abort} removed) is:
First of all, obviously line~\ref{line:CAS:abort} does not change the RMR complexity.
Moreover, if a process receives the abort-signal, then its abortable \NameDecide{} call terminates within a finite number of steps, and the process also does not wait in the while-loop, so its \CAS{} call completes within a finite number of its steps.
Finally, notice that a \CAS{} call returns $\bot$ only if an abort signal was received, and in that case no shared memory objects are affected
(the process cannot have won the \NameDecide{} call).
Hence, all aborted and failed operations can be removed from the execution without changing anything for the remaining operations.
As a result we obtain \cref{thm:CAS_upper_bound}.

\section{RMR Lower Bound for Abortable Leader Election}
In this section, we give an overview of the RMR lower bound proof for abortable leader election (and thus TAS) as stated in \cref{thm:main_lower_bound}.
First, we define some notation, the system model, RMR complexity, and the abortable leader election problem.

\subsection{Lower Bound Preliminaries}
\textbf{System Model and Notation.}\quad
For a set $Q$, set $Q^k$, for some non-negative integer $k$, denotes the set of all sequences of length $k$ that contain only the elements in $Q$. Furthermore, $Q^\ast$ denotes the sets of all sequences that contain only elements of set $Q$.

For the lower bound we assume a set $\mathcal{P}$ of $n$ processes, and an arbitrary large but finite set $\RR$ of shared registers.
Processes are infinite state machines.
In each shared memory step (corresponding to a state transition), a process either reads or writes a register in $\RR$.
At an arbitrary point, a process may also receive an \emph{abort signal} which does not result in a shared memory access, but in a state change of that process, provided the process has not earlier received the abort signal.
Once a process has reached a halting state, it will remain in that state forever, and does not execute any further shared memory steps.
  
For each process $p\in \mathcal{P}$, we define a special abort symbol $p^\top$.
For a set $P\subseteq\mathcal{P}$ let $P^\top=\{p^\top\,|\,p\in P\}$, and $P^\Delta=P\cup P^\top$.
A \emph{configuration} is a sequence that describes the state of each process in $\mathcal{P}$ and each register in $\RR$.
A \emph{schedule} is a sequence $\sigma$ over $\mathcal{P}^\Delta$.
Thus, any schedule $\sigma$ is in $(\mathcal{P}^\Delta)^\ast$.
The length of an schedule $\sigma$ is denoted by $|\sigma|$.
Let $\sigma_1$ and $\sigma_2$ be two schedules. Then $\sigma_1 \circ \sigma_2$ is the schedule obtained by concatenating $\sigma_2$ to the end of $\sigma_1$, without changing the order within $\sigma_1$ and $\sigma_2$.
Let $\Proc(\sigma)$ denote the set of processes $p\in\mathcal{P}$ that occur in $\sigma$ at least once, not counting symbols in $\mathcal{P}^\top$.

A configuration $C$ and a schedule $\sigma\in P^\Delta$ of length one result in a new configuration $Conf(C,\sigma)$, obtained from $C$ by process $p$ taking its next step, if $\sigma=p\in\mathcal{P}$, or by process $p$ receiving the abort signal, if $\sigma=p^\top\in\mathcal{P}^\top$.
If $\sigma=\sigma_1\sigma_2\dots\sigma_k$ is a schedule of length $k>1$, then the new configuration is determined inductively as $\Conf\big(\Conf(C, \sigma_1); \sigma_2\dots\sigma_k \big)$.
Configuration $C$ and schedule $\sigma = \sigma_1\dots\sigma_k$ also define an \emph{execution} $\Exec(C, \sigma)$, which is a sequence $s_1s_2\dots s_k$, where $s_i$ is the step executed or the abort signal received in the transition from $C_{i-1}=\Conf(C, \sigma_1\dots\sigma_{i-1})$ to $C_i = \Conf(C_{i-1}, \sigma_i)$.
To specify that an execution starting in $C$ and running by schedule $\sigma$ is running algorithm $A$, we use $\Exec_A(C, \sigma)$. 
The length of an execution $E$ is denoted by $|E|$.
We call $s_i$ an \emph{abort step by process $p$}, if in $s_i$ process $p$ receives the abort signal.
Let $E_1$ and $E_2$ be two executions. Then $E_1 \circ E_2$ is the execution obtained by concatenating the steps of $E_2$ after the steps of $E_1$, without changing the order of steps within $E_1$ and $E_2$.

The initial configuration is denoted by $\Gamma$.
A configuration $C$ is \emph{reachable}, if there exists a schedule $\sigma$ such that $\Conf(\Gamma, \sigma)=C$. Since only reachable configurations are important in our algorithms and proofs, we use \emph{configuration} instead of \emph{reachable configuration} from this point on. 
For a configuration $C$ we let $\Sched{C}$ denote an arbitrary but unique schedule such that $\Conf(\Gamma, \Sched{C})=C$, and we define $\E{C}=\Exec(\Gamma,\Sched{C})$.

The projection of a schedule $\sigma$ to a set $Q \subseteq\mathcal{P}^\Delta$ is denoted by $\sigma |Q$.
For an execution $E$ and a set $Q$ of processes, $E | Q$ denotes the sub-sequence of $E$ that contains all (abort and shared memory) steps by processes in $Q$.

Recall that a configuration $C$ determines the state of each process.
I.e., for any two executions $E$ and $E'$ resulting in the same configuration $C$, each process is in the same state at the end of $E$ as at the end of $E'$, and in particular $E|p=E'|p$.
Therefore, we associate the state of a process in configuration $C$ with $\E{C}|p$.
(But note that if two executions $E$ and $E'$ are indistinguishable to each process in $Q\subseteq\mathcal{P}$, then this does not in general imply that $E|Q=E|Q'$.)
The value of register $r$ in configuration $C$ is denoted by $\Val_C(r)$.
Configurations $C$ and $D$ are \emph{indistinguishable} to some process $p$, if $\E{C}|p=\E{D}|p$ and $\Val_C(r)=\Val_D(r)$ for every register $r\in\RR$.
For a set $Q\subseteq\mathcal{P}$, we write $C \sim_Q D$ to denote that configurations $C$ and $D$ are indistinguishable to each process in $Q$; for a set consisting of a single process $p$ we write $C \sim_p D$ instead of $C\sim_{\{p\}} D$.

\bigskip
\noindent\textbf{RMR Complexity.}\quad
Our lower bound applies to both, the standard asynchronous distributed shared memory (DSM) model and cache-coherent (CC) model.
In fact, we use a model that combines both, caches as well as locally accessible registers for each process.

We assume that set of registers, $\RR$, is partitioned into disjoint memory segments $\RR_p$, for $p\in\mathcal{P}$.
The registers in $\RR_p$ are \emph{local} to process $p$ and \emph{remote} to each process $q\neq p$.
We say that at the end of execution $E$ a process $p$ has a valid cache copy of register $r$, if in $E$ process $p$ reads or writes $r$ at some point, and no other processes writes $r$ after that. 
Note that the configuration obtained at the end of an execution starting in $\Gamma$ uniquely determines whether $p$ has a valid cache copy of a register $r$.
The reason is that the state of $p$ in configuration $C$ determines the value that was written to or read from $r$ when $p$ accessed $r$ last, and $p$ has a valid cache copy of $r$ if and only if $val_C(r)$ equals that value. 
Let $\Cache_p(C)$ denote the union of $\RR_p$ and the set of registers of which process $p$ has a valid cache copy in configuration $C$ if $p$ has not terminated in $C$, and the empty set if $p$ is terminated in $C$.

A step in an execution $E$ is either \emph{local} or \emph{remote} (we say it \emph{incurs an RMR} if it is remote).
All abort steps are local.
A non-abort step by process $p$ is local, if and only if it is either a read or a write of a register in $\RR_p$, or it is a read of a register of which $p$ has a local cache copy. 

For an execution $E$ and a process $p$, $\RMR_p(E)$ is the number of RMR steps by process $p$ in execution $E$. Further, $\RMR(E)$ is the number of RMR steps incurred by all processes in execution $E$.
For $Q \subseteq \mathcal{P}$ we define $\RMR_Q(E) = \sum_{q \in Q} \RMR_q(E)$, which is equal to the total number of RMRs incurred by processes in $Q$ in $E$. For the sake of conciseness, we use $\RMR(E)$ instead of $\RMR_{\mathcal{P}}(E)$.

\bigskip
\noindent\textbf{Abortable Leader Election.}\quad
An algorithm solves \emph{abortable leader election}, if for any schedule $\sigma$, in $\Exec(\Gamma, \sigma)$ each process that terminates returns \emph{win} or \emph{lose}, at most one process returns \emph{win}, and if all processes in $\Proc(\sigma)$ return \emph{lose}, then all processes in $\Proc(\sigma)$ receive the abort signal.

We usually assume without explicitly saying so that an abortable leader election satisfies \emph{deadlock-freedom} and \emph{bounded abort}, defined as follows:
Bounded abort means that after a process received the abort signal it terminates within a finite number of its own steps.
An infinite execution $\sigma$ is \emph{$P$-fair} for $P \subseteq \mathcal{P}$, if each process appears infinitely many times in $\sigma$.
An infinite execution $E$ is \emph{$P$-fair} for $P\subseteq\mathcal{P}$, if for some configuration $C$ and a $P$-fair schedule $\sigma$, it holds $E = \Exec(C, \sigma)$.
We use \emph{fair} schedule and \emph{fair} execution, instead of $P$-fair, when $P = \mathcal{P}$.
An algorithm is \emph{deadlock-free} if for any schedule $\sigma$ all processes terminate in $Exec(\Gamma,\sigma)$, provided this execution is fair. 

\subsection{Properties of Abortable Leader Election}\label{sec:abortable_LE_properties}
In this section we derive the critical property that distinguishes non-abortable from abortable leader election for the purpose of the lower bound.
We consider algorithms in which each process returns either $win$ or $lose$ upon termination.
We call such algorithms \emph{binary}.
Note that any (abortable) leader election algorithm is a binary algorithm.

Several results in this section will concern only two arbitrarily selected processes in the $n$-process system for $n\geq 2$.
For ease of notation, we will call these processes $a$ and $b$.

For an execution $E$ of a binary algorithm in which $a$ returns $x$ and $b$ returns $y$, let $(x,y)$ denote the \emph{outcome vector} of $E$.
For a binary algorithm $A$ and a configuration $C$, let $\VV_A(C)$ denote the set of all outcome vectors of $\{a, b\}$-only executions starting in $C$, in which processes $a$ and $b$ terminate.

First we observe that the outcome vectors of two indistinguishable configurations are equal.
\begin{observation}
  \label{claim:indSameOutcome}
  For any binary algorithm $A$, if configurations $C$ and $D$ are indistinguishable to processes $a$ and $b$, then $\VV_A(C)=\VV_A(D)$.
\end{observation}
\begin{proof}
  Since $C$ and $D$ are indistinguishable to processes $a$ and $b$, $\E{C} | a = \E{D} | a$, $\E{C} | b = \E{D} | b$, and for any register $r$, $\Val_C(r) = \Val_D(r)$. Thus, for any $x$ in $\{a, b\}^\Delta$, we have $\big(\E{C} \circ \Exec(C, x)\big) | a = \big(\E{D} \circ \Exec(D, x)\big) | a$, $\big(\E{C} \circ \Exec(C, x)\big) | b = \big(\E{D} \circ \Exec(D, x)\big) | b$, and for any register $r$, $\Val_{\Conf(C, x)}(r) = \Val_{\Conf(D, x)}(r)$. So by induction, for any $\{a, b\}$-only schedule $\sigma$, $\Conf(C, \sigma) \stackrel{}{\sim_{\{a, b\}}} Conf(D, \sigma)$.
  Therefore, if in $\Exec(C, \sigma)$ process $p \in \{a, b\}$ terminates, it also terminates in $\Exec(D, \sigma)$ and it returns the same value in both executions.
  Hence, the outcome vector $\VV_A(C)$ is equal to $\VV_A(D)$.
\end{proof}

For a binary algorithm $A$, configuration $C$ is \emph{bivalent} if $\big\{(win, lose), (lose, win)\big\}=\VV_A(C)$.
This definition of bivalency refers to two fixed but arbitrarily chosen processes, $a$ and $b$. 
In a system with more than two processes, we may write $\{a,b\}$-bivalent to indicate the two processes $a$ and $b$ to which this definition applies.
A configuration is \emph{strongly bivalent} (or strongly $\{a,b\}$-bivalent) if it is bivalent and a solo-run by any process $p \in \{a, b\}$, starting in $C$, results in $p$ winning.

A similar argument to the FLP Theorem \cite{FLP1985a} implies that for any deadlock-free binary algorithm and for any reachable bivalent configuration, there exists an infinite execution, where no process terminates.
\begin{lemma}
  \label{lem:infiniteExec}
  Let $A$ be a deadlock-free binary algorithm and $C$ an $\{a,b\}$-bivalent configuration.
  There exists an infinite schedule $\sigma\in \{a,b\}^\ast$, such that in $\Exec_A(C, \sigma)$ none of $a$ and $b$ terminate.
\end{lemma}
To prove this lemma we first prove \cref{claim:FLP} and use the fact that none of $a$ and $b$ can be terminated in an $\{a, b\}$-bivalent configuration.

\begin{claim}
  \label{claim:FLP}
  In any deadlock-free binary algorithm $A$, if configuration $C$ is $\{a, b\}$-bivalent, then either one of $\Conf(C, a)$ and $\Conf(C, b)$ is $\{a, b\}$-bivalent, or there exists an infinite $\{a, b\}$-only execution, where none of $a$ and $b$ terminates. 
\end{claim}
\begin{proof}
  Since configuration $C$ is $\{a, b\}$-bivalent, $\VV_A(C)=\big\{(win, lose), (lose, win)\big\}$.
  Suppose neither $\Conf(C, a)$ nor $\Conf(C, b)$ is $\{a, b\}$-bivalent.
  Then there exist distinct $x, y \in \{win, lose\}$ such that
  \begin{align}
  \label{eq:univalency}
  & \VV_A\big(\Conf(C, a)\big)=\{(x, y)\} \text{ , and} \\
  & \VV_A\big(\Conf(C, b)\big)=\{(y, x)\} \notag
  \end{align}
  We now distinguish two cases.
  
  \emph{Case 1:}\quad In $C$, processes $a$ and $b$ are poised to access different registers or poised to read the same register.
  Thus,
  \begin{equation}
  \label{eq:abOrderDoesntMatter}
  \Conf(C, a \circ b) = \Conf(C, b \circ a).
  \end{equation}
  By \cref{eq:univalency}, $(y, x) \notin \VV_A\big(\Conf(C, a)\big)$. 
  Since $\VV_A\big(\Conf(C, a \circ b)\big)\subseteq \VV_A\big(\Conf(C, a)\big)$), it holds $(y, x) \notin \VV_A\big(\Conf(C, a \circ b)\big)$.
  Thus, by \cref{eq:abOrderDoesntMatter}, $(y, x) \notin \VV_A\big(\Conf(C, b \circ a)\big)$.
  Since $\VV_A\big(\Conf(C, b \circ a)\big)\subseteq \VV_A\big(\Conf(C,  b)\big)=\{(y, x)\}$, this means that $\VV_A\big(\Conf(C, b \circ a)\big)=\emptyset$.
  But this contradicts deadlock-freedom, as in a fair schedule starting in $\Conf(C, b \circ a)$ both processes must terminate and output something.

  \emph{Case 2:}\quad In configuration $C$, both processes are poised to access the same register $r$, and at least one of them is poised to write $r$.
  Without loss of generality, assume that $a$ is poised to write register $r$. If $a$ takes its write step after $b$'s step, then $a$'s state and shared register values are no different than if only $a$ takes its write step and $b$ does not take its step. So $\Conf(C, a) \stackrel{}{\sim_a} \Conf(C, b \circ a)$.
  If process $a$ does not terminate in a solo-run starting in $\Conf(C, a)$, then the claim is true, because there exists an infinite execution starting in $C$ that neither $a$ nor $b$ terminates.
  However, if process $a$ terminates in a solo-run starting in $\Conf(C, a)$, by \cref{eq:univalency}, we can conclude that $(x, y) \in \VV_A\big(\Conf(C, b \circ a)\big)$.
  Since $\VV_A\big(\Conf(C, b \circ a)\big) \subseteq \VV_A\big(\Conf(C, b)\big)$, it holds that $(x, y) \in \VV_A\big(\Conf(C, b)\big)$.
  This contradicts $\VV_A\big(\Conf(C, b)\big) = \{(y, x)\}$.
\end{proof}

Any deadlock-free (non-abortable) 2-process leader election algorithm has a bivalent initial configuration.
But in any fair schedule, both processes terminate.
Therefore, the infinite execution that is guaranteed by the above corollary cannot be fair; in particular, it requires one of the two processes to run solo at some point.
However, one can construct a deadlock-free (non-abortable) leader election algorithm in which one process never takes an infinite number of steps, no matter what the schedule is.
The lemma below shows that this is not true for abortable two-process leader election algorithm.

\begin{lemma}
  \label{lemma:oneWaitFree}
  Let $A$ be a deadlock-free abortable 2-process leader election algorithm $A$ with bounded aborts.
  For any process $p$, there exists an execution starting in the initial configuration, in which $p$ takes an unbounded number of steps.
\end{lemma}
\begin{proof}
  Let $\Gamma$ be the initial configuration of $A$.
  For the purpose of contradiction, assume there is a fixed process, $a$, that terminates within a finite number of its own steps in all executions.
  Let $b$ be the other process.
  
  By the safety property of abortable leader election, there is no execution in which both processes win, i.e., 
  \begin{equation}
  \label{eq:bothCantWin}
  (win, win) \notin \VV_A(\Gamma).
  \end{equation}
  
  Let algorithm $A'$ be the same as $A$ except that during any execution,
  \begin{enumerate}[(1)]
    \item if any of the two processes receive the abort signal, the abort signal is ignored; and
    \item if in step $s$ process $b$ reads $(a, x)$, where $x \neq \bot$, then $b$ continues its program, as if it had received the abort signal immediately after step $s$.
  \end{enumerate}
  
  In any execution of $A'$, $a$ and $b$ can only both lose, if they both receive the abort signal. Since both ignore the abort signals (and only $b$ possibly simulates having received an abort signal), there is no execution of $A'$ in which $a$ and $b$ both lose.
  Thus, for the initial configuration $\Gamma'$ of $A'$,
  \begin{equation}
  \label{eq:loseloseNotInCAprime}
  (lose, lose) \notin \VV_{A'}(\Gamma').
  \end{equation}
  
  Consider any execution $E'=\Exec(\Gamma',\sigma')$ of algorithm $A'$ starting in $\Gamma'$.
  We now create an execution $E=\Exec(\Gamma,\sigma)$ of $A$ starting in $\Gamma$, by scheduling the processes in exactly the same order as in $E'$, but removing all abort signals.
  Moreover, when for the first time $b$ reads a value of $(a, x)$ in $E$, where $x \neq \bot$ (if that happens), then we send process $b$ the abort signal. 
  By construction of $A'$, processes $a$ and $b$ execute exactly the same shared memory steps in execution $E$ of algorithm $A$ as in execution $E'$ of algorithm $A'$.
  Thus, for every schedule $\sigma'$ there is a schedule $\sigma$ such that processes $a$ and $b$ execute in $\Exec_{A'}(\Gamma',\sigma')$ the same shared memory steps as in $\Exec_A(\Gamma,\sigma)$.
  This implies
  \begin{equation}
  \label{eq:AprimeSubsetA}
  \VV_{A'}(\Gamma') \subseteq \VV_A(\Gamma).
  \end{equation}
  Note that in the construction above, if $\sigma'$ is fair, then so is $\sigma$.
  Hence, the fact that $A$ is deadlock-free implies
  \begin{equation}\label{eq:bothCantWint:deadlock-free}
  \text{$A'$ is deadlock-free}.
  \end{equation}
  
  In algorithm $A$, in a sufficiently long solo-run by $a$, in which $a$ does not receive the abort-signal, process $a$ terminates (by deadlock-freedom) and returns $win$ (by the safety property of abortable leader election).
  Hence, in $A'$ process $a$ also terminates and returns $win$ after a sufficiently long solo-run, because it takes exactly the same steps as in $A$.
  Since $A'$ is deadlock-free by \cref{eq:bothCantWint:deadlock-free}, process $b$ terminates after a sufficiently long solo-run following $a$'s solo-run, and by \cref{eq:bothCantWin} process $b$ returns $lose$.
  With a symmetric argument, for algorithm $A'$, in a sufficiently long solo-run by $b$ followed by a sufficiently long solo-run of $a$, process $b$ returns $win$ and process $a$ returns $lose$.
  Hence, $\{(win,lose),(lose,win)\}\subseteq\VV_{A'}(\Gamma')$.
  Using \cref{eq:bothCantWin} and \cref{eq:loseloseNotInCAprime} we conclude
  \begin{equation} 
  \label{eq:CAprimeBivalent}
  \VV_{A'}(\Gamma')=\big\{(win, lose), (lose, win)\big\}.
  \end{equation}
  
  We will now show that $A'$ is wait-free. 
  This together with \cref{eq:CAprimeBivalent} contradicts Lemma~\ref{lem:infiniteExec}, and thus proves the lemma.

  Recall that in every execution of algorithm $A$ process $a$ terminates within a finite number of its own steps.
  As a result, the same is true for $A'$.
  
  Hence, it suffices to show that $b$ terminates within a finite number of its own steps.
  Suppose there is an execution $E^\ast$ of $A'$ in which $b$ executes an infinite number of steps.
  Then $b$ never reads a value of $(a, x)$, where $x \neq \bot$, as otherwise it would simulate having received the abort-signal in $A$, and then terminate after a finite number of steps.
  Since $b$ never reads a value of $(a, x)$, where $x \neq \bot$, it cannot distinguish $E^\ast$ from a solo-run starting in $\Gamma'$.
  Hence, $b$ does not terminate in such an infinite solo-run.
  This contradicts \cref{eq:bothCantWint:deadlock-free}.
\end{proof}

One of the core properties of the abortable leader election problem that allows us to prove the lower bound is that there are no reachable strongly bi-valent configurations in any execution.

\begin{lemma}
  \label{lemma:bothLose}
  Let $A$ be an abortable $n$-process leader election algorithm with bounded aborts for $n\geq 2$.
  Further, let $C$ be a reachable configuration and $a,b$ two distinct processes that terminate in any $\{a, b\}$-fair execution starting in $C$.
  For any schedule $\sigma\in\mathcal{P}^\ast$ configuration $C=\Conf(\Gamma,\sigma)$ is not strongly $\{a, b\}$-bivalent.
\end{lemma}
\begin{proof}
  Suppose $C$ is strongly $\{a, b\}$-bivalent.
  Then it is $\{a, b\}$-bivalent, so
  \begin{equation}
  \label{eq:loseLoseIsNotAPossibleOutcome}
  \VV_A(C)=\{(lose,win),(win,lose)\},
  \end{equation}
  and if $a$ or $b$ runs solo in $C$, then that process wins.
  Because $\sigma \in \mathcal{P}^\ast$, neither $a$ nor $b$ receives the abort-signal in $\Exec(\Gamma,\sigma)$.
  By the assumption that aborts are bounded, processes $a$ and $b$ both terminate in sufficiently long solo runs starting in $\Conf(C, a^\top)$ and $\Conf(C, b^\top)$, respectively.
  Let $x$ and $y$ be the return values of $a$ in $\Exec(C,a^\top\circ a^{k_a})$ and of $b$ in $\Exec(C,b^\top\circ b^{k_b})$, respectively, for sufficiently large integers $k_a$ and $k_b$.
  
  Since $\Conf(C, a^\top) \sim_a \Conf(C, a^\top b^\top)$, 
  \begin{equation}
  \text{$a$ returns $x$ in $\Exec(C,a^\top b^\top\circ a^{k_a})$}.
  \end{equation}
  Similarly, since $\Conf(C, b^\top) \sim_b \Conf(C, a^\top b^\top)$, 
  \begin{equation}
  \text{$b$ returns returns $y$ in $\Exec(C,a^\top b^\top\circ b^{k_b})$}.
  \end{equation}

  We distinguish the following cases.
  
  \emph{Case 1: $x = y = win$:}\quad    
  In a sufficiently long solo-run by $b$ following $\Exec(C,a^\top b^\top\circ a^{k_a})$, process $b$ must terminate (by deadlock-freedom).
  Since $a$ wins in that execution, $b$ must lose.
  Thus,
  \begin{equation}
  \label{eq:winLoseInAAbort}
  (win, lose) \in \VV_A\big(\Conf(C, a^\top b^\top)\big).
  \end{equation}
  Applying a symmetric argument to a sufficiently long solo-run by $a$ following $\Exec(C,b^\top a^\top\circ b^{k_b})$, we obtain
  \begin{equation}
  \label{eq:loseWinInBAbort}
  (lose,win) \in \VV_A\big(\Conf(C, a^\top b^\top)\big).
  \end{equation}
  Hence, using \cref{eq:loseLoseIsNotAPossibleOutcome}, we get $\big\{ (win, lose), (lose, win) \big\}=\VV_A\big(\Conf(C, a^\top b^\top)\big)$.
  Then by Lemma~\ref{lem:infiniteExec}, there exists an infinite execution starting in $\Conf(C, a^\top b^\top)$, such that $a$ and $b$ do not terminate.
  This contradicts bounded aborts.

  \emph{Case 2: $x = y = lose$:}\quad
  In a sufficiently long solo-run by $b$ following $\Exec(C,a^\top b^\top\circ a^{k_a})$, process $b$ must terminate (by deadlock-freedom).
  Since $a$ loses in that execution, by \cref{eq:loseLoseIsNotAPossibleOutcome}, process $b$ must win.
  Thus, $(lose,win)\in\VV_A\big(\Conf(C, a^\top b^\top)\big)$, and with a symmetric argument $(win,lose)\in\VV_A\big(\Conf(C, a^\top b^\top)\big)$.
  We get a contradiction for the same reasons as in Case~1.
  
  \emph{Case 3:}\quad $\{x, y\} = \{win, lose\}$:
  Without loss of generality, assume $x = win$. 
  Then in $\Exec(C,a^\top a^{k_a})$ process $a$ wins.
  On the other hand, since $C$ is strongly bivalent, $b$ wins in a sufficiently long solo-run starting in $C$.  
  Since $C \stackrel{}{\sim_b} \Conf(C, a^\top)$, process $b$ also wins in a long enough solo-run starting in $\Conf(C,a^\top)$. 
  Hence, we have shown that any of the two processes in $\{a,b\}$ wins in a solo-run starting in $\Conf(C,a^\top)$.
  By deadlock-freedom and \cref{eq:loseLoseIsNotAPossibleOutcome} the other process loses, if it performs a long enough solo-run afterwards.
  This shows that $\Conf(C,a^\top)$ is strongly bivalent.
  
  Now let $A'$ be the 2-process algorithm in which $a$ and $b$ act exactly as in algorithm as $A$, but the initial configuration is $\Gamma'=\Conf(C, a^\top)$.
  Then $A'$ is a deadlock-free abortable 2-process leader election algorithm with bounded aborts:
  The bounded abort property is inherited from $A$. 
  Deadlock-freedom follows from the assumption that $a$ and $b$ terminate in any fair execution starting in $C$.
  The safety property of abortable leader election follows from \cref{eq:loseLoseIsNotAPossibleOutcome} and the fact that each process wins in a long enough solo-run starting in the initial configuration $\Conf(C, a^\top)$ (because that configuration is strongly bivalent).
  
  Moreover, in $A'$ process $a$ always terminates within a finite number of its own steps.
  This follows from the bounded abort property of $A$ and the fact that both processes simulate $A$ starting in configuration $\Conf(C,a^\top)$, in which $a$ has already received the abort-signal.
  This contradicts Lemma~\ref{lemma:oneWaitFree}.
\end{proof}

\subsection{Properties of Executions and Safe Configurations}
\subsubsection{Additional Assumptions}
We make the following assumptions that do not restrict the generality of our results.
Recall that processes are state machines, each using some infinite state space $\mathcal{Q}$.
We assume that during an execution a process never enters the same state twice.
Further, we assume that each register stores a pair in $\mathcal{P}\times(\mathcal{Q}\cup\{\bot\})$, where $\bot\notin\mathcal{Q}$.
The initial value of each register in $\RR_p$ is $(p,\bot)$, and when a process $p$ writes to any register, it writes a pair $(p,x)$, where $x$ is $p$'s state before its write operation. 
I.e., we are using a full information model, where processes write all information they have observed in the past. 
As a result, no two writes in an execution write the same value.
Each process's first shared memory step is a read outside of its local shared memory segment, that we call \emph{invocation read}, and thus incurs an RMR. Adding such a step to the beginning of each process's program does not affect the asymptotic RMR complexity of the algorithm.
We will assume that at the end of its execution, each process $p$ reads all registers in $\RR_p$ once.
Since those reads do not incur any RMRs, this assumption can be made without loss of generality.
We call $p$'s last read of register $r\in\RR_p$ the \emph{terminating read} of $r$, and we assume that after $p$'s last terminating read, $p$ will immediately enter a halting state.

%

\subsubsection{Terminology and Notation}
We define some additional terms and notation.
%

We say process $p$ is \emph{visible} on register $r$ in configuration $C$ if $\Val_C(r) = (p, x)$, for some $x \in \mathcal{Q}$.
Let $\Lost(C)$ be the set of processes that have lost in configuration $C$.

When we construct our high RMR execution, we need to make sure that whenever a process gains information about some other process that has not yet lost, someone pays for that with an RMR.
To keep track of who knows who, we define a set $K(C)$ that contains pairs $(p,q)$ of processes.
Informally, $(p,q)$ is in $K(C)$ if $p$  has already gained information about process $q$ in the execution leading to configuration $C$, or $p$ can gain such information for ``free'' (i.e., without an RMR being paid for that).
Gaining information does not only mean that $p$ reads a register that $q$ has written; it means anything that might affect $p$'s execution, e.g., $p$'s cache copies being invalidated.
$K(C)$ is the union of three sets $K_1(C)$, $K_2(C)$, and $K_3(C)$, defined as follows:
\begin{itemize}
  \item $K_1(C)$ is the set of all pairs $(p, q)$, $p \neq q$, such that in $\E{C}$ process $p$ reads a register while process $q$ is visible on that register. I.e., $p$ reads a value of $(q, x)$, where $x \in \mathcal{Q}$.\\
  \emph{Informally:} $p$ has learned about $q$ in $\E{C}$.
  \item $K_2(C)$ is the set of all pairs $(p, q)$, $p \neq q$, such that in $\E{C}$ process $q$ takes at least one shared memory step and process $p$ reads a register in $\RR_q$.\\
  \emph{Informally:} Process $p$ may have a valid cache copy of a register $r\in\RR_q$, and by writing to $r$ process $q$ can invalidate that cache copy without incurring an RMR.
  \item $K_3(C)$ is the set of all pairs $(p, q)$, $p \neq q$, such that in $\E{C}$ process $p$ takes at least one shared memory step, and $q$ writes to a register $r\in\RR_p$ before $p$'s terminating read of $r$.\\
  \emph{Informally:} $p$ may learn about $q$ without incurring an RMR by scanning all its registers in $\RR_p$.
\end{itemize}
Let $\K(C) = K_1(C) \cup K_2(C) \cup K_3(C)$.
We say process $p$ \emph{knows} process $q$ in configuration $C$ if $(p, q) \in \K(C)$.

Recall that in our inductive construction of an RMR expensive execution, we will sometimes erase processes from the constructed execution.
For that reason, if $p$ knows about $q$, i.e., $(p,q)\in K(C)$, then we will not remove a process $q$ from the execution $\E{C}$.
We achieve this by ensuring that whenever $(p,q)\in K(C)$, $q\in L(C)$, and as discussed earlier no lost processes will be erased.

However, we have to be careful about cases in which $p$ does not know directly about $q$.
For example, suppose process $q$ writes to register $r$ in execution $E$, and later some process $z$ overwrites $r$ and finally $p$ becomes poised to read $r$.
In our inductive construction we may want to remove either $z$ or $p$ from the execution, because we do not want $z$ to be discovered by $p$.
However, removing $z$ reveals $q$ on register $r$, and so now $p$ may discover $q$.
To account for that we introduce the concept of \emph{hidden} processes.

In particular, for a configuration $C$ and a register $r$ we define a set $H_r(C)$ of processes \emph{hidden} on $r$ as follows:
\begin{enumerate}[(H1)]
  \item For $r \notin \RR_p$, $p\in H_r(C)$ if and only if either $p$ does not access $r$ in $\E{C}$, or $p$ accesses $r$ in $\E{C}$ at some point $t$, and either no process writes $r$ after $t$, or at least one process that writes $r$ after $t$ is in $\Lost(C)$; \\
  \emph{Idea:} If $p$'s write to $r$ was overwritten by some processes, then at least one of them has lost and thus will not be erased from the execution. Hence, erasing a process does not reveal $p$'s write to any other process.
  \item For $r \in \RR_p$, $p\in H_r(C)$ if and only if any process other than $p$ that writes to $r$ in $\E{C}$ is in $\Lost(C)$.\\
  \emph{Idea:} If a process $q$ wrote to a register $r$ in $p$'s local memory segment, then $q$ has lost. Therefore, $q$ will not be erased from the execution.
  This is important because $p$ can read $r$ for free and we have to assume that it does so frequently, so erasing $q$ from the execution might change what $p$ observes in the execution.
\end{enumerate}
Let $\HH(C) = \bigcap_{r \in \RR} H_r(C)$.
We say process $p$ is \emph{hidden} in configuration $C$, if $p \in \HH(C)$.

We finally define the concept of a safe configuration as follows.
Configuration $C$ is \emph{safe}, if 
\begin{enumerate}[(S1)]
  \item for any pair $(p, q) \in \K(C)$, $q \in \Lost(C)$, and
  \item if $p \notin H(C)$, then either $p \in \Lost(C)$, or $p$ takes no shared memory step in $\E{C}$.
\end{enumerate}
The first property ensures that no process $p$ knows another process $q$ that has not yet lost, and the second property says that all processes that are not hidden must have lost, or not even started participation.
As a result, in an execution leading to a safe configuration, we can erase all processes that do not lose, without affecting any other processes.
Formally, we will prove for a schedule $\sigma$, a safe configuration $C=\Conf(\Gamma, \sigma)$ and a set of processes $P\supseteq \Lost(C)$,
\begin{compactitem}
  \item $\Exec(\Gamma,\sigma)|P=\Exec(\Gamma,\sigma|P^\Delta)$;
  \item $\RMR_P(\Exec(\Gamma,\sigma))=\RMR_P(\Exec(\Gamma,\sigma|P^\Delta))$; and
  \item $\Cache_p(C) =\Cache_p(\Conf(\Gamma,\sigma|P^\Delta))$ for all $p\in P$.
\end{compactitem}
Moreover, if $C$ is safe, then $\Conf(\Gamma, \sigma|P^\Delta)$ is also safe.

\subsubsection{Forcing Processes to Lose}
Lemma~\ref{lemma:bothLose} is a core lemma in the construction of an RMR-expensive execution, which states that we can force two processes to lose starting in a reachable configuration, that the two processes terminate in any fair execution of those two processes and win in their solo execution.
\begin{lemma}
  \label{lem:bothLose}
  Let $C$ be a reachable configuration, and $a, b \in \mathcal{P} \setminus \Lost(C)$ two distinct processes that do not receive the abort signal in $\E{C}$. 
  Further, assume that processes $a$ and $b$ both terminate in any $\{a,b\}$-fair execution starting in $C$.
  If each process in $\{a, b\}$ wins in its solo-run starting in $C$, then there exists a schedule $\sigma \in \big(\{a, b\}^\Delta\big)^\ast$, such that $a$ and $b$ lose in $\Exec(C, \sigma)$.
\end{lemma}
\begin{proof}
  For the purpose of contradiction, assume that for any execution $\Exec(C,\sigma)$, where $\sigma \in \big(\{a, b\}^\Delta\big)^\ast$, in which $a$ and $b$ both terminate, one of the processes wins.
  Then $(lose, lose) \notin \VV_A(C)$.
  Since a solo-run by either $a$ or $b$, starting in $C$, results in that process winning, $C$ is $\{a, b\}$-strongly bivalent. This contradicts Lemma~\ref{lemma:bothLose}.
\end{proof}

\subsubsection{Projections}
We continue by proving properties of the projection operation.
First, the projection of a schedule to a superset of lost processes, $P$, does not change the execution of those processes, if any process that is known by a process in $P$ is lost.
\begin{claim}\label{claim:projection_same_execution}
  Let $\sigma$ be a schedule, $C=\Conf(\Gamma, \sigma)$, and $P\subseteq\mathcal{P}$.
  If $\Lost(C) \subseteq P$, and $q \in \Lost(C)$ for any pair $(p, q)\in\K(C)$, then
  \begin{equation}\label{eq:projection_same_execution}
  \Exec(\Gamma,\sigma)|P=\Exec(\Gamma,\sigma|P^\Delta).
  \end{equation}
\end{claim}
\begin{proof}
  We prove the claim by induction on the length of $\sigma$.
  If $\sigma$ is the empty schedule, then the claim is trivially true.
  
  Now suppose that $\sigma=\sigma'\lambda$, where $\lambda\in\mathcal{P}^\Delta$ is a schedule of length one, and the inductive hypothesis is true for $\sigma'$, i.e.,
  \begin{equation}\label{eq:projection_same_exectuion:IH}
  \Exec(\Gamma, \sigma') | P = \Exec(\Gamma, \sigma' | P^\Delta).
  \end{equation}
  Let $D=\Conf(\Gamma, \sigma')$, and $D'=\Conf(\Gamma, \sigma'|P^\Delta)$.
  We will show that 
  \begin{equation}\label{eq:claim:projection:to_show}
  \Exec(D,\lambda)|P=\Exec(D',\lambda|P^\Delta).
  \end{equation}
  Then it follows from \cref{eq:projection_same_exectuion:IH} that $\Exec(\Gamma,\sigma'\lambda)|P=\Exec(\Gamma,\sigma'\lambda|P^\Delta)$, which completes the inductive step.
  
  If $\lambda\notin P^\Delta$, then each of the two executions on the left and right hand side of \cref{eq:claim:projection:to_show} is the empty execution, so \cref{eq:claim:projection:to_show} is true. 
  Now suppose $\lambda\in\{p,p^\top\}$ for some process $p\in P$.
  Then in $\Exec(D,\lambda)=\Exec(D,\lambda)|P$, either process $p$ receives the abort signal or process $p$ executes a shared memory operation.
  First assume that $p$ receives the abort signal or writes some value $x$ to a shared register $r$ in that step. 
  By \cref{eq:projection_same_exectuion:IH} process $p$ is in the same state in $D$ as in $D'$, so $p$ receives the abort signal or writes $x$ to register $r$, respectively, in $\Exec(D',\lambda)=\Exec(D',\lambda|P^\Delta)$.
  In either case \cref{eq:claim:projection:to_show} follows.
  
  Now assume that in $\Exec(D,\lambda)=\Exec(D,\lambda)|P$, process $p=\lambda$ reads a register $r$.
  Since $p$ is in the same state in $D$ as in $D'$, it reads the same register $r$ in $\Exec(D,\lambda|P^\Delta)$.
  We will show that $\Val_D(r)=\Val_{D'}(r)$.
  As a result, $p$ reads the same value in both executions, and thus \cref{eq:claim:projection:to_show} follows.
  
  For the purpose of a contradiction, assume $\Val_D(r)\neq\Val_{D'}(r)$.
  First assume $\Exec(\Gamma,\sigma')$ contains no write to register $r$.
  Then, by the assumption that $\Val_D(r)\neq\Val_{D'}(r)$, execution $\Exec(\Gamma,\sigma'|P^\Delta)$ contains a write to $r$ by some process $q$.
  Since only processes in $P$ take steps in that execution, $q\in P$. 
  But since $q$ does not write in $\Exec(\Gamma,\sigma')$, we have $\Exec(\Gamma,\sigma')|P \neq\Exec(\Gamma,\sigma'|P^\Delta)$, contradicting \cref{eq:projection_same_exectuion:IH}.
  
  Now assume $\Exec(\Gamma,\sigma')$ contains a write to $r$, and let $w$ be the last such write, executed by some process $q$.
  Thus, $\Val_D(r)=(q,x)$ for some value $x \in \mathcal{Q}$.
  Since in $\Exec(D,\lambda)$ process $p$ reads register $r$, $(p, q) \in K_1\big(\Conf(D, \lambda)\big)$.
  Since $C = \Conf(\Gamma, \sigma) = \Conf(D, \lambda)$, we have $(p, q) \in \K(C)$.
  Therefore, $q \in \Lost(C)$ by the assumption of the claim that $C$ is safe.
  Because $\Lost(C)\subseteq P$, it follows that $q\in P$.
  Therefore, by \cref{eq:projection_same_exectuion:IH}, $q$'s write $w$, with value $(q,x)$, also occurs in $\Exec(\Gamma,\sigma'|P^\Delta)$, and $q$ does not write to $r$ again after $w$.
  By the assumption that $\Val_D(r) \neq \Val_{D'}(r)$, $\Exec(\Gamma,\sigma'|P^\Delta)$ must contain another write $w'$ that is executed after $w$ by some process $q'\neq q$. 
  All steps in that execution are performed by processes in $P$, so $q'\in P$. 
  But then by \cref{eq:projection_same_exectuion:IH}, $w$ and $w'$ are executed in the same order in $\Exec(\Gamma,\sigma')$, contradicting that $w$ is the last write to $r$ in that execution.
\end{proof}

If $C$ is a safe configuration, then by (S1) $q\in\Lost(C)$ for each pair $(p,q)\in\K(C)$.
Hence, from \cref{claim:projection_same_execution} we immediately get:

\begin{corollary}\label{cor:projection_same_execution_from_safe} 
  Let $\sigma$ be a schedule, $C = \Conf(\Gamma, \sigma)$ and $P$ a set of processes such that $\Lost(C) \subseteq P$.
  If $C$ is safe, then 
  \begin{equation}
  \Exec(\Gamma, \sigma)|P = \Exec(\Gamma, \sigma | P^\Delta).
  \end{equation}
\end{corollary}

The projection of a schedule leading to a safe configuration to a superset of lost processes does not change the cached values of those processes.

\begin{claim}
  \label{claim:projection_cache}
  Let $\sigma$ be a schedule, $P\subseteq\mathcal{P}$, $C = \Conf(\Gamma, \sigma)$, and $C'=\Conf(\Gamma, \sigma | P^\Delta)$. 
  If $C$ is safe and $\Lost(C) \subseteq P$, then $\Cache_p(C) = \Cache_p(C')$ for each process $p\in P$.
\end{claim}
\begin{proof} 
  Let $E = \Exec(\Gamma, \sigma)$, and $E' = \Exec(\Gamma, \sigma | P^\Delta)$.  
  Since $C$ is safe, and $\Lost(C) \subseteq P$, by \cref{cor:projection_same_execution_from_safe},
  \begin{equation}\label{eq:clm:projection}
  E|P=E'.
  \end{equation}
  Fix a process $p \in P$. First assume $p \in \Lost(C)$. Thus, since $\Lost(C') \subseteq \Lost(C)$, we have $p \in \Lost(C')$. By definition, $\Cache_p(C) = \Cache_p(C') = \emptyset$.
  
  Now assume $p \notin \Lost(C)$.
  We first show $\Cache_p(C)\subseteq\Cache_p(C')$.
  Let $r\in \Cache_p(C)$.
  Then in some step $s$ of $E$ process $p$ accesses $r$, and no process writes to $r$ after step $s$.
  By \cref{eq:clm:projection}, $p$ also executes step $s$ in $E'$.
  For the purpose of a contradiction assume $r\notin\Cache_p(C')$.
  Then in $E'$ some process $q$ writes to $r$ after step $s$.
  Since only processes in $P$ take steps in $E'$, $q\in P$.
  But then by \cref{eq:clm:projection} process $q$ also writes to $r$ after step $s$ in $E|P$ and thus in $E$---a contradiction.
  
  We now prove $\Cache_p(C')\subseteq\Cache_p(C)$.
  Let $r\in\Cache_p(C')$. If $r \in \RR_p$, then by definition $r \in \Cache_p(C)$. So assume $r \notin \RR_p$.
  Then
  \begin{equation}
  \label{eq:pVisibleOnr}
  \text{in $E'$ process $p$ accesses $r$ and no process writes to $r$ after $p$'s last access.}
  \end{equation}
  By \cref{eq:clm:projection}, $p$ also accesses $r$ in $E|P$, and thus in $E$.
  For the purpose of a contradiction assume $r \notin \Cache_p(C)$.
  Therefore, some process writes to $r$ in $E$ after $p$'s last access of $r$.
  Since $C$ is safe, $p \notin \Lost(C)$, and $p$ takes at least one shared memory step in $\E{C}$, we obtain from (S2) that $p \in \HH(C)$. Thus, by the assumption that $r \notin \RR_p$, by (H1) at least one process, $q$, that writes to $r$ in $E$ after $p$'s last access of $r$, must be in $\Lost(C)$.
  Therefore, $q \in P$. Since $p \in P$, by \cref{eq:clm:projection}, $q$ writes $r$ after $p$'s last access in $E'$. This contradicts \cref{eq:pVisibleOnr}.
\end{proof}

Removing a winning process from a schedule that leads to a safe configuration does not affect the state and cache values of other processes.

\begin{claim}
  \label{claim:winnerProjection}
  Let $\sigma$ be a schedule, such that $C = \Conf(\Gamma, \sigma)$ is safe.
  Further, let $p \in \mathcal{P}$ and $P = \mathcal{P} \setminus \{p\}$.
  If $p$ wins in $\Exec(\Gamma, \sigma)$, then
  $\Exec(\Gamma, \sigma)|P = \Exec(\Gamma, \sigma|P^\Delta)$, and $\Cache_q\big(\Conf(\Gamma, \sigma)\big) = \Cache_q\big(\Conf(\Gamma, \sigma|P)\big)$, for all $q \in P$.
\end{claim}
\begin{proof}
  Because $p$ wins in $\Exec(\Gamma, \sigma)$, we have $\Lost(C) \subseteq P \subseteq \mathcal{P}$. 
  Now the claim follows immediately from tje fact hat $C$ is safe and \cref{cor:projection_same_execution_from_safe,claim:projection_cache}. 
\end{proof}

\subsubsection{Safe Configurations}
The following claims and lemmas describe the properties of safe configurations.
First we show that if starting in a safe configuration, a process that has not yet received the abort signal takes a step which does not incur an RMR, then the resulting configuration is also safe.

\begin{claim}
  \label{claim:zeroRMRsafe}
  Let $C$ be a safe configuration and $x\in\Proc(\Sched{C})$, such that $x^{\top}$ does not appear in $\Sched{C}$.
  If $\RMR\big(\Exec(C, x)\big) = 0$, then $C' = \Conf(C, x)$ is safe.
\end{claim}
\begin{proof} 
  Let $s$ be the single step $\Exec(C, x)$, and $r$ the register accessed in $s$. 
  Since $x$ takes at least one shared memory step in $\E{C}$ (because $x \in \Proc(\Sched{C})$ and $x^\top$ does not appear in $\Sched{C}$), 
  \begin{equation}
  \label{eq:sNotFirstSharedMemoryStep}
  s\text{ is not }x\text{'s first shared memory step in }\E{C} \circ s.
  \end{equation}
  
  Suppose $s$ does not incur an RMR.
  To prove that $C'$ is safe, we will first show that $C'$ satisfies (S1).
  Suppose not.
  Then there exists a pair $(p, q) \in \K(C')$, such that $q \notin \Lost(C')$.
  Since $\Lost(C) \subseteq \Lost(C')$
  \begin{equation}\label{eq:zeroRMRsafe:q_not_lost}
  q \notin \Lost(C).
  \end{equation}
  Since $C$ is safe, $(p, q) \notin \K(C)$, i.e., 
  \begin{equation}\label{eq:zeroRMRsafe:p_gets_to_know_q}
  (p,q)\in K(C')\setminus K(C).
  \end{equation}
  By \cref{claim:KChangeSharedMemoryStep}, 
  \begin{equation}
  \label{eq:xInpq}
  x \in \{p, q\}.
  \end{equation}
  By \cref{eq:zeroRMRsafe:p_gets_to_know_q} there is an index $j\in\{1,2,3\}$ such that $(p,q)\in K_j(C')\setminus K_j(C)$.
  For each of $j\in\{1,2,3\}$ we will show that this is impossible.
  
  If $(p, q) \in K_1(C') \setminus K_1(C)$, then in step $s$ process $p$ reads a register $r$ while process $q$ is visible on $r$. Therefore, the last write to $r$ in $\E{C}$ is by $q$.
  If $r \in \RR_p$, then $(p, q) \in K_3(C)$, which contradicts \cref{eq:zeroRMRsafe:p_gets_to_know_q}. Hence, $r \notin \RR_p$.
  Because $(p, q) \notin K_1(C)$, $p$ does not read $r$ in $\E{C}$ at a point when $q$ is visible on $r$. More specifically, $p$ does not read the value $\Val_C(r)$ from $r$ in $\E{C}$. Thus, in $C$ process $p$ does not have a valid cache copy of $r$. Hence, step $s$ incurs an RMR, which is a contradiction.
  
  Now assume $(p, q) \in K_2(C') \setminus K_2(C)$.
  Since by \cref{eq:sNotFirstSharedMemoryStep}, $s$ is not $q$'s first shared memory step in $\E{C} \circ s$, in step $s$ process $p$ reads $r\in\RR_q$, and $p$ does not read any register in $\RR_q$ throughout $\E{C}$.
  Hence,
  \begin{equation}
  \label{eq:qTakesStepInEC}
  q \text{ takes at least one shared memory step in } \E{C}.
  \end{equation}
  Since $s$ does not incur an RMR, $r \in \Cache_p(C)$, and so $p$ reads or writes $r$ in $\E{C}$, and no other process writes $r$ after that.
  If $p$ reads $r \in \RR_q$ during $\E{C}$, then by \cref{eq:qTakesStepInEC} $(p, q) \in K_2(C)$, which is a contradiction.
  Hence, in $\E{C}$ process $p$ writes $r$, and no other process writes $r$ after that.
  Since $r \in \RR_q$ and $p \notin \Lost(C)$ (as in $C$ process $p$ is poised to executes step $s$), we have $q \notin H_r(C)$ according to (H2), and thus, $q \notin \HH(C)$.
  By \cref{eq:zeroRMRsafe:q_not_lost}, $q \notin \Lost(C)$ and by the claim assumption $q$ takes at least one step in $\E{C}$.
  Therefore, (S2) is not satisfied, which contradicts the assumption that $C$ is safe.
  
  If $(p, q) \in K_3(C') \setminus K_3(C)$, then either $s$ is a write by process $q$ and $r \in \RR_p$, or $s$ is $p$'s first shared memory step. 
  The latter is not possible because of \cref{eq:sNotFirstSharedMemoryStep}.
  And if the former is the case, then $s$ incurs an RMR, which contradicts the assumption that $\RMR\big(\Exec(C, x)\big) = 0$.
  Thus, we have shown that $C'$ satisfies (S1).
  
  We will now prove that $C'$ also satisfies (S2).
  Suppose not.
  Then there exists a process $p \notin \HH(C')$, such that $p \notin \Lost(C')$ and $p$ takes at least one shared memory step in $\E{C'}$.
  Since $\Lost(C)\subseteq \Lost(C')$, we have $p \notin \Lost(C)$.
  
  Recall that $C$ is safe.
  If $p \notin \HH(C)$, then by (S2) process $p$ takes no shared memory steps in $\E{C}$.  
  As $p$ takes a shared memory step in $\E{C'} = \E{C} \circ s$ we have $x = p$, and in particular $s$ is $x$'s first shared memory step.
  This contradicts \cref{eq:sNotFirstSharedMemoryStep}.
  
  If $p \in \HH(C)$, then $p \in \HH(C) \setminus \HH(C')$, which means there exists some register $v$, such that $p \in H_v(C) \setminus H_v(C')$.
  If $v \in \RR_p$, then since $p \notin \HH(C')$, by (H2) in $\E{C'}$ some process $z \notin \Lost(C')$, $z \neq p$, writes to $v$.
  Then $z \notin \Lost(C)$, and so since $p \in \HH(C)$, by (H2) process $z$ does not write $v$ in $\E{C}$.
  Hence, $\Exec(C, x)$ is a write to $v \in \RR_p$ by $z \neq p$, and this write incurs an RMR. 
  This contradicts the claim assumption, $\RMR\big(\Exec(C, x)\big) = 0$.
  
  Now suppose $v \notin \RR_p$.
  Let $q' \neq p$ be the process such that $v \in \RR_{q'}$.
  Because $p \in H_v(C)$, there is a non-empty set $Z$ of processes that write $v$ after $p$'s last access of $v$ during $\E{C}$, and $Z \cap \Lost(C) \neq \emptyset$.
Since $\Lost(C) \subseteq \Lost(C')$, we have  $Z \cap \Lost(C') \neq \emptyset$.
  If step $s$ is not an access of register $v$, $Z$ is also the set of processes that write to $v$ after $p$'s last access of $v$ during $\E{C'}$. So $p$ is in $H_v(C')$.
  If step $s$ is an access of register $v$, then because $\RMR\big(\Exec(C, x)\big) = 0$, process $p$ is not the process performing step $s$.
  Thus, $Z$ is a subset of processes that write to $v$ after $p$'s last access of $v$ during $\E{C'}$. Hence, $p$ is in $H_v(C')$.
\end{proof}

We now show that a process $p$,  which executes a solo-run starting from a safe configuration, must eventually either terminate or incur an RMR.

\begin{claim}
  \label{claim:terminateorRMR}
  Let $C$ be a safe configuration, and let $p$ be an arbitrary process in $\Proc(\Sched{C}) \setminus \Lost(C)$, such that $p^\top$ does not appear in $\Sched{C}$.
  There exists a non-negative integer $k$, such that
  in $\Exec(C, p^k)$, process $p$ terminates or incurs an RMR.
\end{claim}
\begin{proof}
  Assume that there exists a process $p$ that does not terminate and does not incur any RMRs in an infinite solo-run starting in $C$. 
  Let $P = \Lost(C) \cup \{p\}$ and $\sigma = \Sched{C}$. 
  Since $C$ is safe, $p \in \Proc(\Sched{C})$, and $p$ incurs no RMRs in its solo-run starting in $C$, the conditions of \cref{claim:zeroRMRsafe} are met. Hence, for any non-negative integer $t$, by applying \cref{claim:zeroRMRsafe} $t$ times,
  \begin{equation}
  \label{eq:CtIsSafe}
  C_t = \Conf(C, p^t) = \Conf(\Gamma, \sigma \circ p^t) \text{ is safe.}
  \end{equation}
  Since only $p$ takes steps in $\Exec(C, p^t)$, and $p$ does not terminate in its solo-run starting in $C$, we obtain $\Lost(\Conf(C, p^t))=L(C)\subseteq P$.
  This together with 
  \cref{eq:CtIsSafe}
  allows us to apply 
  \cref{cor:projection_same_execution_from_safe} to obtain
  \begin{equation}
  \Exec(\Gamma,\sigma \circ p^t) | P \stackrel{\text{\cref{cor:projection_same_execution_from_safe}}}{=} \Exec\big(\Gamma, (\sigma \circ p^t) | P^\Delta \big) = \Exec\big(\Gamma, (\sigma | P^\Delta) \circ p^t \big).
  \end{equation}
  Therefore, if process $p$ does not terminate or incur any RMRs in its $t$-step solo-run starting in $C$, then $p$ does not terminate or incur any RMRs in its $t$-step solo-run starting in $\Conf(\Gamma, \sigma|P^\Delta)$.
  Since this is true for all $t \geq 0$, in the infinite execution $\Exec(\Gamma, \sigma')$, where $\sigma' = (\sigma | P^\Delta) \circ p \circ p \circ ...$, process $p$ does not terminate.
  But schedule $\sigma'$ is fair, because each process in $\Proc(\sigma')\setminus\{p\}$ is in $L(C)$ and thus loses in $\Exec(\Gamma, \sigma')$, and $p$ performs infinitely many shared memory steps.
  This contradicts deadlock-freedom.
\end{proof}

If at the end of an execution, which starts in a safe configuration, a process that terminates knows the same set of processes as in the beginning of that execution, then that process returns win.

\begin{claim}
  \label{claim:winsStartingInSafe}
  Let $C$ be a safe configuration, $p\in\mathcal{P} \setminus \Lost(C)$, and $\sigma$ a schedule, such that $p^{\top}$ does not appear in $\Sched{C}\circ\sigma$, and
  \begin{equation}
  \label{eq:pKnowsTheSame}
  \text{for any $(p,q)\in\K\bparen{\Conf(C, \sigma)}$ either $q \in \Lost(C)$ or $(p,q)\in\K(C)$}.
  \end{equation}
  If $p$ terminates in $\Exec(C, \sigma)$, then $p$ wins.
\end{claim}
\begin{proof}
  Let $C' = \Conf(C, \sigma)$ and $P = \Lost(C) \cup \{p\}$.
  First note that for any pair $(p,q)\in\K(C')$ either $q \in \Lost(C)$ or $(p,q)\in\K(C)$ by \cref{eq:pKnowsTheSame}, and since $C$ is safe, $q\in\Lost(C)$ according to (S1).
  Thus, we can apply \cref{claim:projection_same_execution} to configuration $C'$ and obtain
  \begin{displaymath}
  \Exec(\Gamma, \Sched{C'} | P^\Delta) = \E{C'} | P.
  \end{displaymath}
  If $p$ terminates in $\Exec(C, \sigma)$, then $p$ also terminates in $\Exec(\Gamma,\Sched{C}\circ\sigma)=\E{C'}$, and thus by the above in $\Exec(\Gamma, \Sched{C'} | P^\Delta)$.
  Thus, it suffices to show that $p$ does not lose in that execution.
  Suppose it does lose.
  Since $\Exec(\Gamma, \Sched{C})$ is a prefix of $\Exec(\Gamma, \Sched{C'} | P^\Delta)$, and all processes in $P\setminus\{p\}$ lose in $\Exec(\Gamma, \Sched{C})$ (we defined $P=\Lost(C) \cup \{p\}$), all processes lose in $\Exec(\Gamma, \Sched{C'} | P^\Delta)$.
  By the safety property of abortable leader election, then all processes that take at least one step in that execution must receive the abort signal.
  In particular, $p$ receives the abort signal in $\Exec(\Gamma, \Sched{C'} | P^\Delta)$, and thus $p^\top$ appears in $\Sched{C'}=\Sched{C}\circ\sigma$.
  This contradicts the claim assumption.
\end{proof}

Starting in a safe configuration, if a process does not get to know any process in its solo execution, then that process wins in its solo-run.

\begin{lemma}
  \label{lemma:winsInaSolorun}
  Let $C$ be a safe configuration, and $p \in \Proc(\Sched{C}) \setminus \Lost(C)$, such that $p^\top$ does not appear in $\Sched{C}$, and
  \begin{equation}
  \label{eq:pDoesntKnowNew}
  \text{for any $k\in \IIN$ and any $(p,q)\in\K\bparen{\Conf(C, p^k)}$ it holds $(p,q)\in\K(C)$.}
  \end{equation}
  Then process $p$ wins in its solo-run starting in $C$.
\end{lemma}
\begin{proof}
  We prove that $p$ terminates in $\Exec(C, p^k)$, for some positive integer $k$. Then by \cref{eq:pDoesntKnowNew}, and \cref{claim:winsStartingInSafe}, $p$ wins in its solo-run starting in $C$, and the lemma follows.
  
  Let $P = \Lost(C) \cup \{p\}$.
  Since $C$ is safe, by \cref{cor:projection_same_execution_from_safe},
  \begin{equation}
  \label{eq:projectionOnP_sameExecution}
  \Exec(\Gamma, \Sched{C} | P^\Delta) = \E{C} | P.
  \end{equation}
  We will show by induction for all $k\geq 0$ that
  \begin{equation}\label{eq:winsInaSolorun:projection}
  \Exec\big(\Gamma, \Sched{C} \circ p^k \big) | p = \Exec\big(\Gamma, (\Sched{C} | P^\Delta) \circ p^k \big) | p.
  \end{equation}
  Note that in $\Exec\big(\Gamma, (\Sched{C} | P^\Delta) \circ p^k \big)$ all processes in $P\setminus\{p\}=\Lost(C)$ lose. 
  Hence, by deadlock-freedom, there is an integer $k_0\in\IIN$ such that $p$ terminates in $\Exec\big(\Gamma, (\Sched{C} | P^\Delta) \circ p^{k_0} \big)$.
  Then by \cref{eq:winsInaSolorun:projection} $p$ also terminates in $\Exec\big(\Gamma, \Sched{C} \circ p^{k_0} \big)$, and by \cref{claim:winsStartingInSafe} it wins in that execution.
  Thus, $p$ wins in a solo-run starting in $C$.
  
  It remains to prove the inductive hypothesis \cref{eq:winsInaSolorun:projection}.
  By \cref{eq:projectionOnP_sameExecution} the hypothesis is true for $k=0$.
  Now assume \cref{eq:winsInaSolorun:projection} is true for some integer $k\geq 0$. 
  Let $x$ be the the last step in $\Exec(\Gamma, \Sched{C} \circ p^{k+1}) $, and $y$ the last step in $\Exec\big(\Gamma, (\Sched{C} | P^\Delta) \circ p^{k+1} \big) $.
  To complete the inductive step, it suffices to show that $x=y$.
  By the inductive hypothesis, $p$ is in the same state in $\Conf(\Gamma, \Sched{C} \circ p^{k})$ as in $\Conf(\Gamma, (\Sched{C} | P^\Delta) \circ p^{k} \big)$.
  Thus, either $x$ and $y$ are both read steps, or they are both write steps, and in the latter case, the value written in step $x$ also gets written in step $y$.
  Thus, if $x$ and $y$ are both write steps, then $x=y$.
  
  Hence, assume $x$ and $y$ are both read steps.
  In that case, $p$ reads the same register $r$ in $x$ as in $y$.
  Let $(a, b)$ be the value $p$ reads in $x$, and $(c, d)$ the value $p$ reads in $y$.
  It suffices to show that $(a,b)=(c,d)$.
  
  First assume that $r$ gets written in the last $k$ steps of $\Exec(\Gamma,\Sched{C}\circ p^k)$.
  Then it must be $p$ that writes $(a,b)$ to $r$ itself (i.e., $a=p$), and by the inductive hypothesis \cref{eq:winsInaSolorun:projection}, $p$ writes the same pair in the last $k$ steps of $\Exec(\Gamma,(\Sched{C}|P^\Delta)\circ p^k)$.
  Moreover, in neither execution it writes to $r$ after writing $(a,b)$ to that register.
  Hence, $(a,b)=(c,d)$.
  
  Now assume that $r$ does not get written in the last $k$ steps of $\Exec(\Gamma,\Sched{C}\circ p^k)$.
  Then by the inductive hypothesis, $r$ does not get written in the last $k$ steps of $\Exec(\Gamma,(\Sched{C}|P^\Delta)\circ p^k)$.
  In particular, $r$ has value $(a,b)$ in configuration $C=\Conf(\Gamma,\Sched{C})$, and value $(c,d)$ in configuration $D=\Conf(\Gamma,(\Sched{C}|P^\Delta))$.
  
  First assume no process writes to $r$ in $\Exec(\Gamma,\Sched{C})$.
  Then by \cref{eq:projectionOnP_sameExecution} no process writes to that register in $\Exec(\Gamma,\Sched{C}|P^\Delta)$, so $(a,b)=(c,d)$ is the initial value of $r$.
  
  Hence, suppose $r$ gets written in $\Exec(\Gamma,\Sched{C})$, and thus the last process writing to $r$ in that execution is $a$.
  Recall that in step $x$ process $p$ reads $(a,b)$ from register $r$, so $(p,a)\in\K_1(C_1)\subseteq\K(C_1)$.
  Then $(p,a)\in\K(C)$ by \cref{eq:pDoesntKnowNew}.
  Since $C$ is safe and by (S1), $a \in \Lost(C)\subseteq P$.
  Since $a\in P$ is the last process to write to $r$ in $\Exec(\Gamma,\Sched{C})$, by \cref{eq:projectionOnP_sameExecution}, it is also the last process to write $r$ in $\Exec(\Gamma,\Sched{C}|P^\Delta)$, and in both executions it writes the value $(a,b)$.
  Hence, $r$ has the same value $(a,b)$ in configuration $C$ as in $D$.
\end{proof}

Starting in a safe configuration, if the executions of two schedules from two disjoint sets of processes do not incur any RMRs, then the execution made up of the concatenation of those schedules does not incur any RMRs and the ordering does not matter.

\begin{claim}
  \label{claim:zeroInformation}
  Let $C$ be a safe configuration, and $Q_0, Q_1 \subseteq \Proc(\Sched{C})$ two disjoint sets of processes, such that for any $j \in \{0, 1\}$ there exists $\sigma_j \in (Q_j^\Delta)^\ast$ with $\RMR\big(\Exec(C, \sigma_j)\big) = 0$. Then
  \begin{enumerate}[(a)]
    \item $\Exec(C, \sigma_0 \circ \sigma_1) | Q_j = \Exec(C, \sigma_j)$, for all $j \in \{0, 1\}$, and
    \item $\RMR\big(\Exec(C, \sigma_0 \circ \sigma_1)\big) = 0$.
  \end{enumerate}
\end{claim}
\begin{proof}
  In $\Exec(C, \sigma_0 \circ \sigma_1)$ all the steps by processes in $Q_0$ are executed before any of the steps by processes in $Q_1$.
  Thus, using $Q_0 \cap Q_1 = \emptyset$, we obtain $\Exec(C, \sigma_0 \circ \sigma_1) | Q_0 = \Exec(C, \sigma_0)$.
  Hence, Part~(a) is true for $j = 0$.
  We now use induction on $|\sigma_1|$ to prove Part~(a) for $j = 1$, as well as to prove Part~(b).
  
  First consider the base case, $|\sigma_1| = 0$.
  Then $\sigma_0 \circ \sigma_1 = \sigma_0$ and
  \begin{equation}
  \label{eq:baseCaseSig0Sig1EqSig0}
  \Exec(C, \sigma_0 \circ \sigma_1) = \Exec(C, \sigma_0).
  \end{equation}
  Therefore, $\Exec(C, \sigma_0 \circ \sigma_1) | Q_0 = \Exec(C, \sigma_0)$.
  Since $Q_0 \cap Q_1 = \emptyset$ and $\sigma_0 \in (Q_0^\Delta)^\ast$, $\Exec(C, \sigma_0) | Q_1$ is the empty execution, which is equal to $\Exec(C, \sigma_1)$. Thus $\Exec(C, \sigma_0 \circ \sigma_1) | Q_1 = \Exec(C, \sigma_1) | Q_1$.
  This proves Part~(a).
  From the claim's assumption $\RMR\big(\Exec(C, \sigma_0)\big) = 0$, and \cref{eq:baseCaseSig0Sig1EqSig0} we obtain $\RMR\big(\Exec(C, \sigma_0 \circ \sigma_1)\big) = 0$.
  This proves Part~(b).
  
  Now suppose $|\sigma_1|>0$, and the inductive hypothesis has been proven for the prefix $\sigma_1'$ of $\sigma_1$ of length $|\sigma_1|-1$.
  I.e., 
  \begin{equation}
  \label{eq:IHPartA}
  \Exec(C, \sigma_0 \circ \sigma_1') | Q_1 = \Exec(C, \sigma_1'); \text{ and}
  \end{equation}
  \begin{equation}
  \label{eq:IHPartB}
  \RMR\big(\Exec(C, \sigma_0 \circ \sigma_1')\big) = 0,
  \end{equation}

  First, assume that $\sigma_1 = \sigma_1' \circ p^\top$, for $p \in Q_1$. Then $\RMR\big(\Exec(C, \sigma_0 \circ \sigma_1)\big) = \RMR\big(\Exec(C, \sigma_0 \circ \sigma_1')\big)$. Thus, by \cref{eq:IHPartB}, Part~(b) is true.
  Moreover, 
  \begin{multline*}
  \Exec(C, \sigma_0 \circ \sigma_1) | Q_1 
  =
  \Exec(C, \sigma_0 \circ \sigma_1' \circ p^\top) | Q_1 
  \stackrel{\cref{eq:IHPartA}}{=}
  \Exec\bparen{C,((\sigma_0 \circ \sigma_1')|Q_1)\circ p^\top}
  \\ =
  \Exec(C,\sigma_1' \circ p^\top)
  =
  \Exec(C,\sigma_1).
  \end{multline*}
  This proves Part~(a) for $j=1$.
  
  
  Now assume $\sigma_1 = \sigma_1' \circ p$, for $p \in Q_1$.
  Let $s$ be the last step in $\Exec(C, \sigma_0 \circ \sigma_1)$, and $s'$ the last step in $\Exec(C, \sigma_1)$. 
  We will show:
  \begin{align}\label{eq:zeroInformation:to_show_s=s'}
  &s=s';\ \text{and}\\
  \label{eq:zeroInformation:to_show_no_RMR}
  &\text{step $s$ incurs no RMR in execution $\Exec(C,\sigma_0 \circ \sigma_1)=\Exec(C,\sigma_0 \circ \sigma_1')\circ s$}.
  \end{align}
  Then Part~(b) follows immediately from \cref{eq:IHPartB,eq:zeroInformation:to_show_no_RMR}, and Part~(a) for $j=1$ from
  \begin{multline*}
  \Exec(C, \sigma_0 \circ \sigma_1) | Q_1 
  = \bparen{\Exec(C,\sigma_0 \circ \sigma_1')|Q_1}\circ s
  \stackrel{\cref{eq:IHPartA}}{=}
  \Exec(C,\sigma_1')\circ s
  \\ \stackrel{\cref{eq:zeroInformation:to_show_s=s'}}{=}
  \Exec(C,\sigma_1')\circ s'
  = \Exec(C,\sigma_1).
  \end{multline*}
  
  First note that using \cref{eq:IHPartA} and because $p \in Q_1$ we have
  \begin{equation}\label{eq:zeroInformation:p_same_state}
  \text{in $\Conf(C, \sigma_1')$ process $p$ is in the same state as in $\Conf(C,\sigma_0 \circ \sigma_1')$}.
  \end{equation}
  We separately consider the case that $s$ is a read and that $s$ is a write. 
  
  \emph{Case 1: Step $s$ is a write}.\quad 
  By \cref{eq:zeroInformation:p_same_state} process $p$ writes the same value to the same register in $s$ as in $s'$. 
  This implies \cref{eq:zeroInformation:to_show_s=s'}.
  Moreover, 
  \begin{displaymath}
  \RMR\big(\Exec(C, \sigma_1') \circ s' \big) = \RMR\big(\Exec(C, \sigma_1' \circ p)\big) = \RMR\big(\Exec(C, \sigma_1)\big) = 0,
  \end{displaymath}
  where the last equality follows from the claim's assumption.
  Hence, $s'$ does not incur an RMR, which is only possible if in $s'$ process $p$ writes a register in $\RR_p$.
  Because $s=s'$, $s$ does not incur an RMR either, and so \cref{eq:zeroInformation:to_show_no_RMR} follows.
  
  \emph{Case 2: Step $s$ is a read.}\quad
  Let $r$ be the register process $p$ reads in step $s$, and thus by \cref{eq:zeroInformation:p_same_state}, also in $s'$.
  
  We first prove \cref{eq:zeroInformation:to_show_s=s'}.
  To that end we will show that the value of $r$ is the same in $\Conf(C, \sigma_0 \circ \sigma_1')$ as in $\Conf(C, \sigma_1')$.
  As a result, in step $s$ process $p$ reads the same value from $r$ as in step $s'$, and so $s=s'$.
  
  All writes to $r$ in $\Exec(C, \sigma_1')$ are by processes in $Q_1$ and thus they occur also in $\Exec(C, \sigma_0 \circ \sigma_1')$ in the same order.
  Hence if there is a write to $r$ in $\Exec(C, \sigma_1')$, then the value at the end of $\Exec(C, \sigma_1')$ is the same as at the end of $\Exec(C, \sigma_0 \circ \sigma_1')$.
  In that case $p$ reads the same value in $s$ as in $s'$.
  
  Therefore, assume that $r$ does not get written in $\Exec(C,\sigma_1')$.
  If it also does not get written in $\Exec(C, \sigma_0 \circ \sigma_1')$, then $r$ has the same value at the end of both executions, and $p$ reads that value in both, $s$ and $s'$.
  So suppose $r$ gets written in $\Exec(C, \sigma_0)$ but not in $\Exec(C, \sigma_0 \circ \sigma_1')$, and for the last time it gets written by a process $q$.
  Then $q\in Q_0$, and since $\RMR\big(\Exec(C, \sigma_0)\big)=0$, $r\in\RR_q$.
  
  Since $p\in Q_1$, we have $p\neq q$, and thus $r\not\in\RR_p$.
  Process $p$ reads $r$ during $\Exec(C, \sigma_1)$ at least once (in its last step $s$).
  By the claim's assumption no such read by $p$ incurs an RMR, so $r \in \Cache_p(C)$.
  But then in $\E{C}$ process $p$ reads or writes register $r\in\RR_q$ before $q$'s terminating read (because $q$ writes $r$ in $\Exec(C,\sigma_0)$.
  If $p$ reads $r$ in $\E{C}$, then $(p, q) \in \K_2(C)$, and if $p$ writes $r$ in $\E{C}$, then $(q,p)\in \K_3(C)$.
  Hence, we have either $(p,q)\in \K(C)$ or $(q,p)\in\K(C)$.
  Since $C$ is safe, (S1) implies either $q \in \Lost(C)$ or $p \in \Lost(C)$.
  But neither is possible, as $q$ takes a step in $\Exec(C,\sigma_0)$ (its write to $r$) and $p$ a step in $\Exec(C,\sigma_1)$ (step $s$).
  This is a contradiction, and completes the proof of \cref{eq:zeroInformation:to_show_s=s'}.
  
  Thus, it remains to show \cref{eq:zeroInformation:to_show_no_RMR}, i.e., that $s$ incurs no RMR in $\Exec(C, \sigma_0 \circ \sigma_1')\circ s$.
  If $r\in\RR_p$, then this is obviously true, so assume $r \notin \RR_p$.
  Since $s'=s$ does not incur an RMR in $\Exec(C, \sigma_1') \circ s'$ process $p$ reads $r$ during $\Exec(C, \sigma_1')$, and $r$ does not get written afterwards.
  By \cref{eq:IHPartA} the same is true in $\Exec(C,\sigma_0 \circ \sigma_1')$.
  Hence, at the end of that execution $p$ has a valid cache copy of $r$, so $s$ does not incur an RMR in $\Exec(C,\sigma_0 \circ \sigma_1')\circ s$.
\end{proof}

Starting in a safe configuration, if a process terminates without incurring any RMR steps, it does not gain information and hence wins.
\begin{claim}
  \label{claim:zeroRMROneWins}
  Let $C$ be a safe configuration, and $p$ a process in $\Proc(\Sched{C}) \setminus \Lost(C)$, such that $p^\top$ does not appear in $\Sched{C}$. If $p$ terminates without incurring any RMRs in $E = \Exec(C, p^k)$, for some positive integer $k$, then $p$ wins in $E$.
\end{claim}
\begin{proof}
  Let $C' = \Conf(C, p^{k'})$, for arbitrary $k' \in \{1, ..., k\}$.
  Because $p$ is the only process that takes steps in $E$, it is true that $\big( K_3(C') \setminus K_3(C) \big) \cap (\{p\} \times \mathcal{P}) = \emptyset$ (remember that $K_3(C)$ is the set of all pairs $(a, b)$, $a \neq b$, such that in $\E{C}$ process $a$ takes at least one shared memory step, and $b$ writes to a register $r \in \RR_a$ before $a$'s terminating read of $r$).
  Since $p$ does not incur any RMRs in $E$, if $p$ reads some register $r$ during $E$, then either $r \in \RR_p$, or $r \in \Cache_p(C)$.
  Thus, $\big( K_2(C') \setminus K_2(C) \big) \cap (\{p\} \times \mathcal{P}) = \emptyset$, and $\big( K_1(C') \setminus K_1(C) \big) \cap (\{p\} \times \mathcal{P}) = \emptyset$.
  Hence, $\Big( \K(C) \setminus \K\big(\Conf(C, p^{k'})\big) \Big) \cap (\{p\} \times \mathcal{P}) = \emptyset$, for any $k' \in \{1, ..., k\}$.
  Thus, by Lemma \ref{lemma:winsInaSolorun}, $p$ wins in $E$.
\end{proof}

As long as the set of knowing relations does not change during an execution starting from a safe configuration, at most one process terminates.

\begin{claim}
  \label{claim:oneTerminates}
  Let $C$ be a safe configuration, such that if $p^\top \in \mathcal{P}^\top$ appears in $\Sched{C}$, then $p \in \Lost(C)$.
  Then for any schedule $\sigma \in \mathcal{P}^\ast$, when $\K(C) = \K\big(\Conf(C, \sigma)\big)$, at most one process terminates in $\Exec(C, \sigma)$.
\end{claim}
\begin{proof}
  Let $\sigma \in \mathcal{P}^\ast$, such that $\K(C) = \K\big(\Conf(C, \sigma)\big)$.
  Assume that in $E = \Exec(C, \sigma)$ two distinct processes, $p$ and $q$, terminate.
  Since we assumed that $p$ terminates in $E$, process $p$ is not terminated in $C$, and hence, $p \in \mathcal{P} \setminus \Lost(C)$.
  Because $\K(C) = \K\big(\Conf(C, \sigma)\big)$, the set $\K(C) \setminus \K\big(\Conf(C, \sigma)\big) \cap (\{p\} \times \mathcal{P}) = \emptyset$.
  Further, by the claim statement, $p^\top$ does not appear in $\Sched{C}$ and $\sigma$.
  Thus, by \cref{claim:winsStartingInSafe}, $p$ wins in $\Exec(C, \sigma)$, and by symmetry, $q$ wins in $\Exec(C, \sigma)$. This contradicts the safety property of abortable leader election.
\end{proof}

Projecting a schedule, that leads to a safe configuration, to a superset of all lost processes leads to a safe configuration.

\begin{claim}
  \label{claim:projectionSafe}
  Let $\sigma$ be a schedule, such that $C = \Conf(\Gamma, \sigma)$ is safe. Let $P$ be a set of processes, such that $\Lost(C) \subseteq P \subseteq \Proc(\sigma)$. Then $C' = \Conf(\Gamma, \sigma | P^\Delta)$ is safe.
\end{claim}
\begin{proof}
  For the purpose of contradiction assume that $C'$ is not safe.
  First assume there exists a process $p \notin \HH(C')$, such that $p$ takes at least one shared memory step in $\E{C'}$ and $p \notin \Lost(C')$.
  Because $p$ takes at least one shared memory step in $\E{C'}$, $p \in P$.
  Since $C$ is safe, for any pair $(p, q) \in \K(C)$, process $q$ is in $\Lost(C)$.
  Hence, by \cref{claim:projection_same_execution}, $\Exec(\Gamma, \sigma) | P = \Exec(\Gamma, \sigma | P^\Delta)$.
  Therefore, $p$ takes at least one shared memory step in $\E{C}$ and $p \notin \Lost(C)$.
  Because $p \notin \HH(C')$, there exists a register $r \in \RR$, such that $p \notin \HH_r(C')$.
  
  If $r \in \RR_p$, then at least one process that writes to $r$ in $\E{C'}$, is not in $\Lost(C')$.
  Let $q$ be one of the processes that write to $r$ in $\E{C'}$ and are not in $\Lost(C')$. Since $q$ takes a step in $\E{C'}$, process $q$ is in $P$, and by \cref{claim:projection_same_execution}, takes the same write step to $r$ and is not in $\Lost(C)$.
  Therefore, $p \notin \HH_r(C)$, which contradicts $C$ being safe.
  
  If $r \notin \RR_p$, then in $\E{C'}$ process $p$ writes to $r$, and at least one process, $q$, writes to $r$ after $p$'s write, such that $q \notin \Lost(C')$.
  Since $q$ takes a step in $\E{C'}$, process $q$ is in $P$, and by \cref{claim:projection_same_execution}, takes the same write step to $r$ and is not in $\Lost(C)$.
  Therefore, $p \notin \HH_r(C)$, which contradicts $C$ being safe.
  
  Now assume that for any $p \notin \HH(C')$, either $p \in \Lost(C')$ or $p$ does not take any shared memory steps in $\E{C'}$. Then there exists a pair $(p, q) \in \K(C') \setminus \K(C)$, such that $q \notin \Lost(C')$.
  
  If $(p, q) \in K_1(C') \setminus K_1(C)$, then both $p$ and $q$ take steps in $\E{C'}$ ($p$ takes at least a read step, and $q$ takes at least a write step), and thus, are in $P$.
  If $(p, q) \in K_2(C') \setminus K_2(C)$, then both $p$ and $q$ take steps in $\E{C'}$ ($p$ takes at least a read step, and $q$ takes at least a shared memory step), and thus, are in $P$.
  If $(p, q) \in K_3(C') \setminus K_3(C)$, then both $p$ and $q$ take steps in $\E{C'}$ ($p$ takes at least a shared memory step, and $q$ takes at least a write step), and thus, are in $P$.
  Hence, by \cref{claim:projection_same_execution}, $p$ and $q$ take the same steps in $\E{C}$ and $\E{C'}$. This contradicts $(p, q) \in K(C') \setminus K(C)$.
  
\end{proof}

\subsubsection{Auxiliary Claims}

We now show that during an execution, the knowing relations can only change as a result of a shared memory step by one of the processes, that is in the difference of the relation sets.

\begin{claim}
  \label{claim:KChangeSharedMemoryStep}
  Let $\sigma \in \mathcal{P}^\Delta$, $C$ a configuration, and $C' = \Conf(C, \sigma)$.
  If there exists a pair $(p,q)$ in the symmetric set difference of $\K(C')$ and $\K(C)$, then $\Exec(C, \sigma)$ is a shared memory step by $p$ or by $q$.
\end{claim}
\begin{proof}
  Let $s=\Exec(C, \sigma)$, and $(p, q)$ be a pair in the symmetric set difference of $\K(C)$ and $\K(C')$.
  Step $s$ causes the difference between $K_1(C) \cup K_2(C) \cup K_3(C)$ and $K_1(C') \cup K_2(C') \cup K_3(C')$.
  If $K_1(C) \neq K_1(C')$, then in step $s$ process $p$ reads a register on which $q$ is visible.
  If $K_2(C) \neq K_2(C')$, then either $s$ is $q$'s first shared memory step, or in $s$ process $p$ reads a register in $\RR_q$. 
  Finally, if $K_3(C) \neq K_3(C')$, then $s$ is $p$'s first shared memory step, or in $s$ process $q$ writes to a register in $r\in\RR_p$.
  In all cases, $s$ is a shared memory step by $p$ or $q$.
\end{proof}

If two executions are equal when projected to a set of processes, $P$, then each process in $P$ takes the same number of RMR steps and knows the same set of processes in $P$ at the end of the execution.

\begin{claim}
  Let $P$ be a set of processes, $\sigma$ and $\sigma'$ schedules, and define $E=\Exec(\Gamma, \sigma)$, $E'=\Exec(\Gamma, \sigma')$, $C=\Conf(\Gamma, \sigma)$, and $C'=\Conf(\Gamma, \sigma')$.
  If $E|P=E'|P$, then
  \begin{enumerate}[(a)]
    \item $\RMR_p(E) = \RMR_p(E')$, for any process $p \in P$, and
    \item $\K(C) \cap (P \times P) = \K(C') \cap (P \times P)$.
  \end{enumerate}
\end{claim}
\begin{proof}
  Recall that we assume without loss of generality, that a value does not get written twice in the same execution. Hence, if $p$ reads a value $v$ from register $r$ in execution $E$, then that read incurs no RMR if and only if $p$ accessed $r$ earlier, and in its preceding access of $r$ process $p$ either read or wrote the same value $v$. Therefore, $E|p$ uniquely determines which of $p$'s steps are RMRs, and in particular $\RMR_p(E)$. This proves Part~(a).
  
  We will show that $\K(C) \cap (P \times P) \subseteq \K(C') \cap (P \times P)$. By symmetry, this implies $\K(C') \cap (P \times P) \subseteq \K(C) \cap (P \times P)$, and thus Part~(b).
  Let $(a,b) \in \K(C) \cap (P \times P)$. Then $a, b \in P$, and $(a, b) \in K_1(C) \cup K_2(C) \cup K_3(C)$.
  
  If $(a, b) \in K_1(C)$, then in some step of execution $E$ process $a$ reads a value of $(b, x)$, where $x \in \mathcal{Q}$, from some register $r$. Since $E|P = E'|P$, in $E'$ process $a$ reads $(b, x)$ from $r$. Thus, $(a, b) \in \K(C')$.
  
  If $(a, b) \in K_2(C)$, then in $E$ process $a$ reads a register $r \in \RR_b$, and $b$ takes at least one shared memory step. As $E|P = E'|P$, process $b$ takes at least one shared memory step in $E'$ and $a$ reads $r$ in $E'$. Therefore, $(a, b) \in \K(C')$.
  
  If $(a, b) \in K_3(C)$, then in $E$ process $a$ takes at least one shared memory step, and $b$ writes a register, $r \in \RR_a$, before $a$'s terminating read of $r$. Since $E|P = E'|P$, process $a$ takes at least one shared memory step in $E'$, and $b$ writes $r$ in $E'$, before $a$'s terminating read of $r$. Hence, $(a, b) \in \K(C')$.
  
  Thus, $\K(C) \cap (P \times P) \subseteq \K(C') \cap (P \times P)$.
\end{proof}

\subsection{Constructing an RMR-Expensive Execution}\label{sec:construction}
We now consider an abortable leader election algorithm.
We will construct a schedule such that in an execution starting in the initial configuration at least one process takes $\Omega(\log n / \log\log n)$ RMR steps, where $n$ is the number of processes.

\subsubsection{Overview of the Construction}
Let $n \geq 4$, $\ell =\floor{\log n / c \log\log n}$ for some sufficiently large constant $c$ (which we determine in the appendix).
We inductively construct a schedule $\sigma_i$ and a set of processes $P_i \subseteq \mathcal{P}$, for all $i \in \{0, ..., \ell\}$.
For the sake of conciseness, let $E_i = \Exec(\Gamma, \sigma_i)$, $C_i = \Conf(\Gamma, \sigma_i)$, and $L_i = \Lost(C_i)$.

The construction will satisfy the following invariants for $i \in \{0, ..., \ell\}$:
\begin{compactenum}[(\text{I}1)]
  \item $C_i$ is safe.
  \item $|P_i \setminus L_i| \geq (n-1) / (\log n)^{ci}$.
  \item $\RMR_{P_i \setminus L_i}(C_i) \geq i\,|P_i \setminus L_i| - i$.
  \item For each process $p \in P_i \setminus L_i: \RMR_p(C_i) \leq i$.
  \item For each process $p \in P_i \setminus L_i$, $p^\top$ does not appear in $\sigma_i$.
\end{compactenum}
Invariant  (I2)  for $i=\ell$ implies $|P_\ell\setminus L_\ell|\geq 2$.
Hence, by (I3) there are at least two processes that each incur $\Omega(\ell)=\Omega(\log n/\log\log n)$ RMRs.
\Cref{thm:main_lower_bound} follows.

We now sketch how we construct $\sigma_i$ and $P_i$ inductively so that the invariants are satisfied.
We start with $P_0=\mathcal{P}$ and the initial configuration $C_0$.
We then schedule processes in rounds.
In round $i$, we choose a subset $P_{i+1}$ of the processes in $P_i\setminus L_i$ and remove all processes in $\mathcal{P}\setminus(P_{i+1}\cup L_i)$ from the execution constructed so far.
This does not affect any of the remaining processes, because $C_i$ is safe.
Then we schedule the processes in $P_{i+1}$ in such a way that each of them incurs an RMR, and only a small fraction of them lose.

To decide which processes to remove and to schedule the remaining processes, we proceed as follows:
First we let each process in $P_i\setminus L_i$ take sufficiently many steps until it is poised to incur an RMR.
It is not hard to see that in an execution in which no process incurs an RMR, processes do not learn about each other, so the resulting configuration, $D_i$, is again safe.
Moreover, in a safe configuration processes only know about lost processes, so they cannot lose.

We then distinguish between a high contention write case, where a majority of processes are poised to write to few registers, and a low contention write case, where either many registers are poised to being accessed or a majority of processes are poised to read.
Let $S_i$ be the set of registers processes in $P_i\setminus L_i$ are poised to access in configuration $D_i$.
The high contention write case occurs if there are few such registers and a majority of processes are poised to write, i.e., $|S_i| = O(|P_i\setminus L_i|/\log n)$, and otherwise the low contention write case occurs.

In the low contention write case, we choose a set $Q_i$ of processes, which contains for each register $r\in S_i$ at most one process poised to write to $r$ in $D_i$.
We consider the step $s_p$ each process $p\in Q_i$ is poised to take.
We then create a directed graph $G$ with processes as vertices, and an edge from $p$ to $q$ if in the resulting configuration
(I) due to $s_p$ or $s_q$ process $p$ knows $q$, or 
(II) due to step $s_p$ process $q$ is not hidden.
Each application of rule (I) must be paid for by RMRs in the execution, and for each application of (II) a process $p$ must overwrite some process $q$.
As a result graph $G$ is sufficiently spares, and by Tur\'an's theorem \cite{Tur1941a} we obtain a large independent set $J$.
We let each process $p\in J$ take one step, $s_p$, and erase all remaining processes that haven't lost yet from the execution.
It is not hard to see that no process loses in any of the steps added, the resulting configuration is safe (this follows from how we added edges to $G$) and, because of the sparsity of the graph, a sufficiently large number of processes survive.
From that we obtain Invariants (I1) and (I2).
Since each process $p$ performs an RMR in step $s_p$ and only local steps before that, we get (I3) and (I4).
Moreover, we don't abort any processes, so (I5) is true.

In the high contention write case, we erase all readers from the execution.
For each register $r\in S_i$, let $W_r$ denote the set of processes poised to write to $r$.
Since this is a high contention case, $|W_r|$ is large for most registers $r$.
For each register $r$ with sufficiently large $|W_r|$, we choose two distinct processes $a, b \in W_r$.

We then argue that, after erasing some $O(\log n)$ processes, we obtain a configuration $D_i'$ and an $\{a,b\}$-only schedule $\sigma$ such that in execution $\Exec(D_i',\sigma)$ processes $a$ and $b$ both lose and see no process other than those in $L_i$, which have lost already.
The argument is based on Lemma~\ref{lemma:bothLose}, but quite involved.
We now let, starting from $D_i'$, all processes in $W_r \setminus \{a, b\}$ execute one step, in which they write to $r$.
After that we schedule $a$ and $b$ as prescribed by $\sigma$.
Then $a$ and $b$ will both first write to $r$, and thus overwrite the writes by all other processes in $W_r$, then continue to take steps and lose without seeing any processes that haven't lost, yet.
As a result, all processes in $W_r \setminus \{a, b\}$ have taken a step but are now hidden, two processes ($a$ and $b$) have lost, and $O(\log n)$ processes have been removed.
It is not hard to see that the resulting configuration is safe again.
We repeat this for all registers $r$ for which $|W_r|$ is large enough.
Then, we let $P_{i+1}$ denote the set of all surviving processes and $C_{i+1}$ the resulting configuration.

Configuration $C_{i+1}$ is safe, and sufficiently few processes are removed or have lost so that (I1) and (I2) remain true.
Moreover, each process that does not lose performs exactly one RMR, so (I3) and (I4) are true.
(I5) is true because all processes that received the abort signal lost.

\subsubsection{Partial Execution Constructions}
One of the critical properties that results in constructing a long enough execution, is that we can keep many processes running while keeping them from gaining information. What follows are the formal description and proofs of this property.

First, we claim that the information exchanged during specific executions is bounded.
\begin{claim}
\label{claim:boundsOnKandM}
Let $C$ be a safe configuration, $P = \Proc(\Sched{C}) \setminus \Lost(C)$, and $\sigma \in P^\ast$, such that in $C$ each process in $P$ is poised to perform an RMR step, and in $\Exec(C, \sigma)$ each process takes at most one step and each register gets written at most once. Then
\begin{enumerate}[(a)]
\item $|\K\big(\Conf(C, \sigma)\big) \cap (P \times P)| \leq 2\RMR\big(\Exec(C, \sigma)\big)$.
\item Let $M$ be the set of pairs $(p, q) \in (P \times P)$, $p \neq q$, such that in $\Exec(C, \sigma)$, process $q$ writes a register in $\RR_p \cup \Cache_p(C)$. Then $|M| \leq \RMR(\E{C}) + \RMR\big(\Exec(C, \sigma)\big)$.
\end{enumerate}
\end{claim}
\begin{proof}
Since $C$ is safe, by (S1), $\K(C) \cap (P \times P)$ is the empty set.
Thus, to prove Part~(a) it is sufficient to show that each step in $\Exec(C, \sigma)$ adds at most two pairs of processes to $\Big(\K\big(\Conf(C, \sigma)\big) \setminus \K(C)\Big) \cap (P \times P)$.
Let $\sigma'$ be a proper prefix of $\sigma$, and $p$ a process so that $\sigma' \circ p$ is also a prefix of $\sigma$.
Since $p$'s state is the same in $\Conf(C, \sigma')$ as in $C$, and $p$ is poised to perform an RMR step in $C$, the step $\Exec\big(\Conf(C, \sigma'), p\big)$ incurs an RMR. Now let $C_1 = \Conf(C, \sigma')$ and $C_2 = \Conf(C, \sigma' \circ p)$.

First assume step $\Exec(C_1, p)$ is a read from some register $r \in \RR_{q_2}$, $q_2 \in P$. Let $(q_1, x) = \Val_{C_1}(r)$ (if $r$ is in its initial state, then $x = \bot$ and $q_1 = q_2$).
We prove that no pair other than $(p, q_1)$ and $(p, q_2)$ is in $\K(C_2) \setminus \K(C_1)$.
Suppose $(p', q') \in \K(C_2) \setminus \K(C_1)$, $p', q' \in P$. Hence, $(p', q')$ is in one of the sets $K_1(C_2) \setminus K_1(C_1)$, $K_2(C_2) \setminus K_2(C_1)$, and $K_3(C_2) \setminus K_3(C_1)$.
If $(p', q') \in K_1(C_2) \setminus K_1(C_1)$, then in $\Exec(C_1, p)$ process $p'$ reads a register on which $q'$ is visible. Since $p$ takes the step $\Exec(C_1, p)$ and only $q_1$ can be the process visible on $r$, we have $p = p'$ and $q' = q_1$.
If $(p', q') \in K_2(C_2) \setminus K_2(C_1)$, then since $q'$ takes at least one shared memory step in $\E{C_1}$, in $\Exec(C_1, p)$ process $p'$ reads a register in $\RR_{q'}$. Since $p$ takes the step $\Exec(C_1, p)$ and $r \in \RR_{q_2}$, we have $p = p'$ and $q' = q_2$.
If $(p', q') \in K_3(C_2) \setminus K_3(C_1)$, then since $p'$ takes at least one shared memory step in $\E{C_1}$, process $q'$ writes a register in $\RR_{p'}$ during $\Exec(C_1, p)$. This contradicts $\Exec(C_1, p)$ being a read step.

Now assume step $\Exec(C_1, p)$ is a write to register $r \in \RR_q$.
We prove no pair other than $(q, p)$ is in $\K(C_2) \setminus \K(C_1)$.
Suppose $(q', p') \in \K(C_2) \setminus \K(C_1)$.
Hence, $(q', p')$ is in one of the sets $K_1(C_2) \setminus K_1(C_1)$, $K_2(C_2) \setminus K_2(C_1)$, or $K_3(C_2) \setminus K_3(C_1)$.
Since $\Exec(C_1, p)$ is a write step, no process reads a register in that step and thus, $(q', p') \notin K_1(C_2) \setminus K_1(C_1)$.
If $(q', p') \in K_2(C_2) \setminus K_2(C_1)$, then in $\E{C_2}$ process $p'$ takes at least one shared memory step and $q'$ reads a register in $\RR_{p'}$.
Since $\Exec(C_1, p)$ is a write step, it must be the first shared memory step by $p'$ and $p = p'$. This contradicts $p \in P$.
If $(q', p') \in K_3(C_2) \setminus K_3(C_1)$, then in $\Exec(C_1, p)$ process $p'$ writes a register in $\RR_{q'}$. Thus, $p' = p$, and since $r \in \RR_q$, we have $q' = q$.
Therefore, $(q, p)$ is the only pair in $\K(C_2) \setminus \K(C_1)$.

In order to prove Part~(b), we map each pair in $M$ to an RMR step in $\E{C} \circ \Exec(C, \sigma)$ in such a way that the mapping is injective.
Consider a pair $(p, q) \in M$. I.e., during $\Exec(C, \sigma)$, process $q$ writes to a register $r \in \RR_p \cup \Cache_p(C)$.
If $r \in \RR_p$, then we map $(p, q)$ to $q$'s write step to $r$. Recall that in $\Exec(C, \sigma)$ each process executes at most one step, and that step incurs an RMR. So $(p, q)$ is mapped to a unique RMR step.
Now suppose $r \notin \RR_p$, so $r \in \Cache_p(C)$.
Then there exists a step in $\E{C}$ or in $\Exec(C, \sigma)$, prior to $q$'s write, in which $p$ caches $r$. Let $(p, q)$ be mapped to the last such step. 
That step incurs an RMR, so it suffices to show that the mapping is injective.
First note that if $(p, q)$ is mapped to a step $s$, then in its unique step in $\Exec(C, \sigma)$ process $q$ writes to the register that is accessed in step $s$.
Suppose two distinct pairs, $(p_1, q_1)$ and $(p_2, q_2)$ are mapped to the same step $s$. Let $r$ be the register accessed in $s$.
Then in their steps in $\Exec(C, \sigma)$, processes $q_1$ and $q_2$ must both write to $r$.
Since only one process writes to $r$ during $\Exec(C, \sigma)$, we have $q_1 = q_2$. Therefore, $p_1 \neq p_2$, and so $r \notin \RR_{p_j}$ for some $j \in \{1, 2\}$. Without loss of generality assume $j = 1$.
Then $r \in \Cache_{p_1}(C)$, and step $s$ is by $p_1$.
If $r \notin \RR_{p_2}$, then $(p_2, q)$ would not be mapped to $s$ (it would be mapped to a step by $p_2$).
Thus, $r \in \RR_{p_2}$, so $(p_2, q_2)$ is mapped to $q_2$'s step in $\Exec(C, \sigma)$. This means that step $s$ is performed by process $q_2$.
Hence, $p_2 = q_2$, which contradicts the definition of $M$.
\end{proof}

Then, we construct and prove the properties of an execution where we have a low-contention write case (where either most processes are poised to read, or many registers are poised to being accessed).
\begin{lemma}
\label{lemma:LowContentionAndRead}
Let $\ell$ be a positive integer, $C$ a safe configuration, and $P = \Proc(\Sched{C}) \setminus \Lost(C)$, such that in $\E{C}$ each process in $P$ takes at most $\ell$ RMR steps and does not receive the abort signal, and in $C$ each process in $P$ is poised to perform an RMR step. If in $C$ at least half of the processes in $P$ are poised to read or at least $|P|/(10\ell)$ different registers are poised to being accessed by processes in $P$, then there exists a set of processes $Q \subseteq P$ and a schedule $\sigma \in \big((Q \cup \Lost(C))^\Delta\big)^\ast$, such that
\begin{enumerate}[(a)]
\item $|Q| \geq |P|/(60\ell^2) - 1$,
\item $\Conf(\Gamma, \sigma)$ is safe,
\item $\RMR_Q\big(\Exec(\Gamma, \sigma)\big) = \RMR_Q(\E{C}) + |Q|$, and
\item no process in $Q$ receives the abort signal in $\Exec(\Gamma, \sigma)$.
\end{enumerate}
\end{lemma}
\begin{proof}
Let $V = \{x_1, ..., x_m\}$ be a maximal subset of $P$ such that for each register $r$, set $V$ contains none of the processes that are poised to write to $r$ in $C$, or $V$ contains at most one process that is poised to access $r$ in configuration $C$.
Let $C' = \Conf\Big(\Gamma, \Sched{C} | \big(V \cup \Lost(C)\big)^\Delta \Big)$.
Hence, by \cref{cor:projection_same_execution_from_safe} processes in $V$ are in the same state in $C'$ as they are in $C$, and by \cref{claim:projection_cache} have the same cache.
Therefore, all processes in $V$ are poised to perform an RMR step and access the same registers in $C'$ as in $C$.
Create a directed graph $G$, where each process in $V$ forms a vertex, and where there is an edge from $p$ to $q$, $p \neq q$, if one of the following is true:
\begin{enumerate}[(i)]
\item in $\Conf(C', x_1 \circ ... \circ x_m)$, process $p$ knows process $q$ \Big(I.e. $(p, q) \in \K\big(\Conf(C', x_1 \circ ... \circ x_m)\big)$\Big); or
\item in $\Exec(C', x_1 \circ ... \circ x_m)$, process $q$ writes to a register $r \in \RR_p \cup \Cache_p(C')$.
\end{enumerate}
Let $M$ be the set of edges in $G$ because of (ii).
Since each process in $V$ is poised to perform an RMR step in $C'$, each process takes at most one step, and each register gets written at most once in $\Exec(C', x_1 \circ ... \circ x_m)$, by \cref{claim:boundsOnKandM} Part~(a), the number of edges in $G$ from condition (i) is at most $2\RMR\big(\Exec(C', x_1 \circ ... \circ x_m)\big)$.
From \cref{claim:boundsOnKandM} Part~(b), the number of edges in $G$ from condition (ii) is $|M| \leq \RMR\big(\Exec(C', x_1 \circ ... \circ x_m)\big) + \RMR(\E{C'})$.
Let $Q'$ be a largest independent set in graph $G$, where the direction of edges are ignored.

By \cref{cor:projection_same_execution_from_safe}, $\E{C} | \big(Q' \cup \Lost(C)\big) = \Exec\Big(\Gamma, \Sched{C} | \big(Q' \cup \Lost(C)\big)^\Delta\Big)$.
Further, since no two processes in $Q'$ satisfy condition (i), by \cref{claim:winsStartingInSafe} if a process terminates in $\Exec\big(C', (x_1 \circ ... \circ x_m) | Q'\big)$ it wins.
Let $X$ be the set containing any process that terminates in $\Exec\big(C', (x_1 \circ ... \circ x_m) | Q'\big)$.
Let $Q = Q' \setminus X$, $O = Q \cup \Lost(C')$, and $\sigma = \big(\Sched{C'} | (O^\Delta)^\ast\big) \circ \big((x_1 \circ ... \circ x_m) | Q\big)$.
Because at most one process wins in a leader election algorithm $|X| \leq 1$, and thus, $|Q| \geq |Q'|-1$.

By Tur\'an's theorem\cite{Tur1941a}, the size of the largest independent set in a graph with average degree $d$ and $k$ vertices, is at least $k/(d+1)$.
The number of edges in $G$ is at most
\begin{equation}
\begin{split}
2\RMR\big(\Exec(C', x_1 \circ ... \circ x_m)\big) + \RMR\big(\Exec(C', x_1 \circ ... \circ x_m)\big) + \RMR(\E{C'}) = \\
3\RMR\big(\Exec(C', x_1 \circ ... \circ x_m)\big) + \RMR(\E{C'}).
\end{split}
\end{equation}
Since each step in $\Exec(C', x_1 \circ ... \circ x_m)$ incurs an RMR and each process takes at most $\ell$ RMR steps during $\E{C'}$, the number of edges in $G$ is $3m + m\ell$.
Because $|V| = m$, the average degree of $G$ is at most $2(3m + m\ell)/m$.
Hence, the size of $Q'$ is at least
\begin{equation}
\label{eq:sizeofQ}
\frac{m}{\frac{2(3m + m\ell)}{m} + 1} = \frac{m}{6 + 2\ell + 1} = \frac{m}{7 + 2\ell}.
\end{equation}
The assumption is that in $C$ either at least $|P|/2$ processes are poised to read, or at least $|P|/(10\ell)$ registers are poised to being accessed.
Hence, $m \geq \min\big\{|P|/2, |P|/(10\ell)\big\} \stackrel{\ell \geq 2}{=} |P|/(10\ell)$ and so by \cref{eq:sizeofQ}
\begin{equation}
|Q'| \geq \frac{m}{7 + 2\ell} \geq \frac{|P|}{(7 + 2\ell)10\ell} \stackrel{\ell \geq 2}{\geq} \frac{|P|}{(4\ell + 2\ell)10\ell} = \frac{|P|}{60\ell^2}.
\end{equation}
Since $|Q| \geq |Q'|-1$, Part~(a) is proven.

First, we observe that $C'$ is safe by \cref{claim:projectionSafe}.
Hence, each process $p \in Q$, we have $\E{C'}|p = \Exec\big(\Gamma, \Sched{C'} | (Q \cup \Lost(C'))^\Delta \big) | p$ (this is true by $C'$ being safe and \cref{cor:projection_same_execution_from_safe}).
Further by \cref{claim:projection_cache}, process $p$ has the same cache in $\Conf\big(\Gamma, \Sched{C'} | (Q \cup \Lost(C'))^\Delta\big)$ as in $C'$.
Hence, each process in $Q$ is poised to take the exact same step that incurs an RMR in $\Conf\big(\Gamma, \Sched{C'} | (Q \cup \Lost(C'))^\Delta\big)$.
Therefore, since no two processes that satisfy (i) or (ii) are in $Q$, we have
\begin{equation}
\label{eq:sameRMRStepQ}
\Exec(C', x_1 \circ ... \circ x_m) | Q = \Exec\Big(\Conf\big(\Gamma, \Sched{C'} | (Q \cup \Lost(C'))^\Delta \big), x_1 \circ ... \circ x_m\Big) | Q.
\end{equation}

Let $C'' = \Conf\big(\Gamma, \Sched{C'} | O^\Delta \big)$. By \cref{claim:projectionSafe}, $C''$ is safe.
Let $D = \Conf(\Gamma, \sigma)$ and $E = \Exec(\Gamma, \sigma)$. Remember that $\sigma = \big(\Sched{C} | O^\Delta \big) \circ \big((x_1 \circ ... \circ x_m) | Q\big)$.
Let $E' = \Exec\big(C'', (x_1 \circ ... \circ x_m) | O^\Delta \big)$.
We prove Part~(b) by contradiction.
Assume that $D$ is not safe.
Hence, at least one of (S1) or (S2) is violated.
First assume that (S1) is not true for $D$. Thus, there exists a pair $(p, q) \in \K(D)$, such that $q \notin \Lost(D)$.
Then $q \notin \Lost(C') \subseteq \Lost(D)$, and since $C''$ is safe, we have $(p, q) \notin \K(C'')$. Hence, $p$ gets to know $q$ in $E'$.
From \cref{eq:sameRMRStepQ} and (i), there is an edge between $p$ and $q$ in $G$, which contradicts $p$ and $q$ both being in an independent set of graph $G$.
Now assume that (S2) is not true for $D$.
Hence, there exists a process $p \in O$, such that $p \notin \HH(D)$, $p \notin \Lost(D)$, and $p$ takes at least one shared memory step in $\E{D}$.
Since $C''$ is safe, (S2) is true for $C''$.
Because $p \in \Proc(C'') \setminus \Lost(C'')$, process $p$ takes at least one shared memory step in $\E{C''}$. Hence, because process $p \notin \Lost(D)$, it holds $p \in \HH(C'')$.
Since any process that takes at least one shared memory step in $\E{D}$ and is not in $\Lost(D)$ is in $Q$, it holds $p \in Q$.
Hence, since there is no edge from $p$ to any process in $Q$, by \cref{eq:sameRMRStepQ} and (i) no process in $Q \setminus \{p\}$ writes to a register in $\RR_p$ during $E'$.
Thus, if (S2) is not satisfied for $D$, then (H2) is true for $p$.
Therefore, the reason that $p$ is not hidden in $D$ is because of (H1).
Hence, there exists a register $r \notin \RR_p$, such that process $p$ accesses $r$ in $\E{D}$ at some point $t$ and at least one other process writes to $r$ after $t$, but none of the processes that write to $r$ after $t$ are in $\Lost(D)$.
If $p$'s last access to $r$ is in $\E{C''}$, then since $C''$ is safe and $\Lost(C'') \subseteq \Lost(D)$, all the processes that write to $r$ after $t$ write during $E'$.
Therefore, $r \in \Cache_{C''}(p)$, and if there exists a process $q \in Q$ that is poised to write to $r$ in $C''$, then there is an edge from $p$ to $q$ in $G$, which contradicts $p, q \in Q$.
If $p$'s last access to $r$ is during $E'$, then no other process writes $r$ after that.
This, completes the proof of Part~(b).

Since each process in $Q$ has the same cache in $C''$ as in $C$ and in $E'$ each process in $Q$ takes the step that it is poised to take in $C$, each process in $Q$ performs an RMR step in $E'$. Hence, $E'$ incurs $|Q|$ RMRs, which proves Part~(c).

Since processes in $Q$ do not receive the abort signal in $\E{C''}$, and no process receives the abort signal in $E'$, Part~(d) is true.
\end{proof}

For a high-contention write case on a specific register (where many processes are poised to write to it), we present a way to construct an execution that can be used to construct our desirable execution.

\begin{claim}
\label{claim:both_writers_can_lose_starting_in_safe}
Let $C$ be a safe configuration, such that for a fixed register $r$ each process in $P_r \subseteq \Proc(\Sched{C}) \setminus \Lost(C)$ is poised to perform an RMR write step to $r$ in $C$, for any execution $E$ starting in $C$, no process incurs more than $\ell$ RMRs during $\E{C} \circ E$, and any process that receives the abort signal in $\E{C}$ is in $\Lost(C)$.
There exists a set of processes $Q \subseteq \Proc(\Sched{C}) \setminus \Lost(C)$, and a schedule $\sigma \in (P_r^\Delta)^\ast$, such that for configuration $C' = \Conf\big(\Gamma, \Sched{C} | (Q \cup \Lost(C))^\Delta \big)$, 
\begin{enumerate}[(a)]
\item $\Conf(C', \sigma)$ is safe,
\item $|Q| \geq |\Proc(\Sched{C}) \setminus \Lost(C)| - (8\ell - 1)$,
\item in $\Exec(C', \sigma)$ each process in $(Q \cap P_r) \setminus \Lost\big(\Conf(C', \sigma)\big)$ takes exactly one RMR step,
\item any process that receives the abort signal in $\Exec(C', \sigma)$ is in $\Lost\big(\Conf(C', \sigma)\big)$, and
\item $|\Lost\big(\Conf(C', \sigma)\big) \setminus \Lost(C)| \leq 2$.
\end{enumerate}
\end{claim}
\begin{proof}
By \cref{cor:projection_same_execution_from_safe}, processes in $Q$ are in the same state in $C'$ as in $C$.
Further, by \cref{claim:projectionSafe}, configuration $C'$ is safe.

If $|P_r| < 8\ell - 1$, then let $Q = \Proc(\Sched{C}) \setminus P_r$, and $\sigma$ be the empty schedule.
Since $\sigma$ is the empty schedule, $\Conf(C', \sigma) = C'$, which is a safe configuration. This proves Part~(a).
Because $Q = \Proc(\Sched{C}) \setminus P_r$, we have
\begin{equation}
Q \setminus \Lost\big(\Conf(C', \sigma)\big) = \Proc(\Sched{C}) \setminus \Big(\Lost\big(\Conf(C', \sigma)\big) \cup P_r \Big) =
\Proc(\Sched{C}) \setminus \big(\Lost(C) \cup P_r \big).
\end{equation}
This means
\begin{equation}
\begin{split}
|Q \setminus \Lost\big(\Conf(C', \sigma)\big)| \geq |\Proc(\Sched{C}) \setminus \Lost(C)| - |P_r| \stackrel{|P_r| < 8\ell-1}{\geq}\\
|\Proc(\Sched{C}) \setminus \Lost(C)| - (8\ell - 1),
\end{split}
\end{equation}
which proves Part~(b).
Since $Q \cap P_r = \emptyset$, Part~(c) is true.
Further, no process receives the abort signal in $\Exec(C', \sigma)$, which proves Parts~(d) and (e).

Now suppose $|P_r| \geq 8\ell - 1$.
Let $P = \Proc(\Sched{C}) \setminus \Lost(C)$.
For each process $p \in P_r$, let $Z_p \subseteq P \setminus \{p\}$ be the set of all processes $q$, such that process $p$ reads a register $r'$ in its solo-run starting in $C$, where either $r' \in \RR_q$ or in $C$ process $q$ is visible on $r'$.
Note that because $C$ is safe and $p \in P_r \subseteq P$, for any process $q$ visible on a register in $\RR_p$, we have $(p, q) \in K_3(C)$. Hence, none of the processes in $Z_p$ are visible on any register in $\RR_p$. 
Since $C$ is safe and $Z_p \cap \Lost(C) = \emptyset$, for any process $q \in Z_p$, we have $(p, q) \notin \K(C)$.
Hence, $p$ does not have a cache copy of any register that $q$ is visible on (otherwise, $(p, q) \in K_1(C)$) or any register in $\RR_q$ (otherwise, $(p, q) \in K_2(C)$).
Thus, in a solo-run by $p$ starting in $C$ the first read from a register on which $q$ is visible or a register in $\RR_q$ incurs an RMR.
Hence, since each process incurs at most $\ell$ RMRs during any execution, $|Z_p| \leq 2\ell$, for any $p \in P_r$.
We want to choose two processes $a$ and $b$ from $P_r$, such that $a \notin Z_b$ and $b \notin Z_a$.
We have $\binom{|P_r|}{2}$ many possibilities to choose 2 processes. However, for each process $p$ at most $|Z_p|$ many choices need to be removed.
Therefore, by
\begin{equation}
\binom{|P_r|}{2} - |P_r| \cdot \max_{p \in P_r}\{|Z_p|\} \stackrel{|P_r| \geq 8\ell-1, |Z_p| \leq 2\ell}{\geq} (8\ell - 1)(8\ell - 2)/2 - (8\ell - 1)2\ell = 16\ell^2 - 10\ell + 2 \stackrel{\ell \geq 1}{\geq} 8 > 1,
\end{equation}
we have at least one pair of processes $a$ and $b$, such that $b \notin Z_a$ and $a \notin Z_b$.
Fix a pair of processes $a$ and $b$, such that $b \notin Z_a$ and $a \notin Z_b$.
Let $D = \Conf\Big(\Gamma, \Sched{C} | \big(\Lost(C) \cup \{a, b\} \big)^\Delta \Big)$.
By \cref{claim:projectionSafe}, configuration $D$ is safe.
Hence, because for any pair $(p, q) \in \K\big(\Conf(D, p^k)\big)$, for any $p \in \{a, b\}$ and any positive integer $k$, we have $q \in \Lost(D)$, by \cref{claim:winsStartingInSafe}, in a solo-run by $p$ starting in $D$, in which $p$ does not receive the abort signal $p$ wins.
Hence, a solo run by $p \in \{a, b\}$, in which $p$ does not receive the abort signal, starting in $\Conf(D, a)$ also results in $p$ winning.
That is because when $p = a$, we have $\Exec(D, a^k) = \Exec\big(\Conf(D, a), a^{k-1} \big)$, for any positive integer $k$, and when $p = b$, in $\Exec\big(\Conf(D, a), b^{k'}\big)$, for any positive integer $k'$, the value written by process $a$ is overwritten by $b$ and thus, $D$ and $\Conf(D, a)$ are indistinguishable to process $b$.
Since $a$ and $b$ have not received the abort signal in $\Conf(D, a)$ and in any fair execution starting in $\Conf(D, a)$ both $a$ and $b$ terminate (because the algorithm that we are running is deadlock-free), by Lemma \ref{lem:bothLose},
\begin{equation}
\label{eq:lambdaBothLose}
\text{there exists a schedule } \lambda \in (\{a, b\}^\Delta)^\ast \text{, such that in } \Exec\big(\Conf(D, a), \lambda\big) \text{ both $a$ and $b$ lose.}
\end{equation}
Let $R$ be the set of registers that are being read during $\Exec\big(\Conf(D, a), \lambda\big)$ by any process in $\{a, b\}$.
Further, let $Y \subseteq P \setminus \{a, b\}$ be a set of all processes $q$, such that $q$ is visible on at least one register in $R$ in configuration $C$ or $R \cap \RR_q \neq \emptyset$.
Since $C$ is safe, for any process $y \in Y$, we have $(a, y) \notin \K(C)$ and $(b, y) \notin \K(C)$.
Thus, $\big(\Cache_a(C) \cup \Cache_b(C)\big) \cap R = \emptyset$.
Hence, in any $\{a, b\}$-only execution starting in $C$, for each register $r \in R$ on which a process $q \in Y$ is visible or $r \in \RR_q$, the first read by each process in $\{a, b\}$ from $r$ incurs an RMR. Thus, because $a$ and $b$ incur at most $\ell$ RMRs in $\Exec\big(\Conf(C, a), \lambda\big)$, it is true that $|Y| \leq 4\ell$.

Let $w$ be the process, such that $r \in \RR_w$.
Note that since all processes in $P_r$ are poised to perform an RMR step on $r$, we have $w \notin P_r$.
Further let $X = Z_a \cup Z_b \cup Y \cup \{w\}$, and $D' = \Conf\big(\Gamma, \Sched{C} | (\Lost(C) \cup (P \setminus X))^\Delta \big)$.
Since $C$ is safe, by \cref{cor:projection_same_execution_from_safe}, processes in $P_r \setminus X$ are in the same state in $D'$ as they are in $C$, which means they are poised to write to $r$.
Let $\{q_1, ..., q_k\} = P_r \setminus (X \cup \{a, b\})$, $Q = (P \cup \{a, b\}) \setminus X$, and $\sigma = q_1 \circ ... \circ q_k \circ a \circ \lambda$.
Configurations $\Conf(D, a)$, $\Conf(C', a)$, and $\Conf(C', q_1 \circ ... \circ q_k \circ a)$ are indistinguishable to processes $a$ and $b$.
Hence, by \cref{eq:lambdaBothLose},
\begin{equation}
\label{eq:abLose}
a, b \in \Lost\big(\Conf(C', \sigma)\big).
\end{equation}

We now show that (S1) and (S2) are satisfied for $\Conf(C', \sigma)$.
For (S1) we need to show for any pair $(p, q) \in \K\big(\Conf(C', \sigma)\big)$, that $q \in \Lost\big(\Conf(C', \sigma)\big)$.
Fix a pair $(p, q) \in \K\big(\Conf(C', \sigma)\big)$.
If $(p, q) \in \K(C')$, then since $C'$ is safe, $q \in \Lost(C')$. Because $\Lost(C') \subseteq \Lost\big(\Conf(C', sigma)\big)$, we have $q \in \Lost\big(\Conf(C', sigma)\big)$.
If $(p, q) \notin \K(C')$, then since any visible process on a register read by $a$ or $b$ in $\Exec(C', \sigma)$ is lost (otherwise, it is a process in $Y$, which does not take any steps in $\E{C'} \circ \Exec(C', \sigma)$), we have $(p, q) \notin K_1\big(\Conf(C', \sigma)\big) \setminus K_1(C')$.
Further, because $P \subseteq \Proc(\sigma)$, no process takes its first shared memory step in $\Exec(C', \sigma)$.
Hence, $(p, q) \notin K_2\big(\Conf(C', \sigma)\big) \setminus K_2(C')$.
Thus, $(p, q) \in K_3\big(\Conf(C', \sigma)\big) \setminus K_3(C')$.
Therefore, since $Q \subseteq \Proc(\sigma)$, during $\Exec(C', \sigma)$ process $q$ writes to a register in $\RR_p$.
Because of \cref{eq:abLose}, if $q \in \{a, b\}$, then $q \in \Lost\big(\Conf(C', \sigma)\big)$.
Since each process in $P_r$ is poised to preform an RMR step in $C'$ and both $a$ and $b$ are poised to write to $r$ in $C'$, we have $r \notin \RR_a \cup \RR_b$.
Thus, $p \notin \{a, b\}$.
Since for the process $w$ that $r \in \RR_w$, it holds $w \notin Q$, we have $p \notin Q$.
Hence, (S1) is satisfied.
Since $C'$ is safe, for any $p \notin \HH(C')$, either $p$ does not take any shared memory steps in $\E{C'}$, or $p \in \Lost(C')$.
Thus, because any register that is accessed in $\Exec(C', \sigma)$ is last accessed by either $a$ or $b$, and by \cref{eq:abLose}, (S2) is also satisfied.
This proves Part~(a).

From $Q = (P \cup \{a, b\}) \setminus X$ we get
\begin{equation}
|Q| \geq |\Proc(\Sched{C}) \setminus \Lost(C)| + 2 - |X|.
\end{equation}
Thus, to prove Part~(b) is suffices to prove $|X| \leq 8\ell + 1$.
As $X = Z_a \cup Z_b \cup Y \cup \{w\}$, where $w$ is the process that $r \in \RR_w$, it holds $|X| \leq |Z_a| + |Z_b| + |Y| + 1 \leq 8\ell + 1$.

In $\Exec(C', \sigma)$, each process in $\{q_1, ..., q_k\}$ takes a single write step to $r$.
Since $(P_r \cap Q) \setminus \Lost\big(\Conf(C', \sigma)\big) = \{q_1, ..., q_k\}$, Part~(c) is true.

Processes $a$ and $b$ are the only processes that receive the abort signal in $\Exec(C', \sigma)$, and they both lose.
Thus, Parts~(d) and (e) are true.
\end{proof}

We use the execution constructed in \cref{claim:both_writers_can_lose_starting_in_safe} to create an execution to handle the high-contention write case (where most processes are poised to write and few registers are poised to being accessed).

\begin{lemma}
\label{lemma:high_contention_write_case}
Let $C$ be a safe configuration and $P = \Proc(\Sched{C}) \setminus \Lost(C)$, such that in $\E{C}$ each process in $P$ does not receive the abort signal, and is poised to perform an RMR step.
Also no process takes more than $\ell$ RMR steps in any execution.
If in $C$ more than $|P|/2$ processes are poised to write and at most $|P|/(10\ell)$ registers are poised to being accessed, then there exists a set of processes $Q \subseteq P$ and schedule $\sigma \in \Big( \big(Q \cup \Lost(C)\big)^\Delta \Big)^\ast$, such that
\begin{enumerate}[(a)]
\item configuration $C' = \Conf(\Gamma, \sigma)$ is safe,
\item $|Q \setminus \Lost(C')| \geq |P|/10$,
\item $\RMR_{Q \setminus \Lost(C')}\big(\Exec(\Gamma, \sigma)\big) = \RMR_{Q \setminus \Lost(C')}(\E{C}) + |Q \setminus \Lost(C')|$, and
\item any process that receives the abort signal in $\Exec(\Gamma, \sigma)$ is in $\Lost(C')$.
\end{enumerate}
\end{lemma}
\begin{proof}
Let $P_0 \subseteq P$ be the set of processes that are poised to write in configuration $C$, and $C_0 = \Conf\Big(\Gamma, \Sched{C} | \big(P_0 \cup \Lost(C)\big)^\Delta \Big)$.
Let $\{r_1, ..., r_k\}$ be the set of registers that are poised to being written in $C$.
We inductively construct schedule $\sigma_i$, for $i \in \{1, ..., k\}$.
Our inductive hypothesis is that for $i \in \{1, ..., k\}$,
\begin{enumerate}[(\text{IH}1)]
\item configuration $C_i = \Conf(\Gamma, sigma_i)$ is safe,
\item $|P_i| \geq |P| - i(8 \ell - 1)$,
\item any process that receives the abort signal in $\E{C_i}$ is in $\Lost(C_i)$.
\end{enumerate}
Since $C$ is safe, by \cref{claim:projectionSafe}, configuration $C_0$ is safe.
By \cref{claim:projection_same_execution}, it holds $\E{C}|P = \E{C_0}|P$.
Therefore, any process that receives the abort signal in $\E{C_0}$ is in $\Lost(C_0)$.
Thus, by (IH3) it holds that any process in $\Proc(C_i) \setminus \Lost(C_i)$, for $i \in \{0, ..., k\}$, does not receive the abort signal in $\E{C_i}$.
Hence, by (IH1) and the fact that no process takes more than $\ell$ RMR steps in any execution starting in $\Gamma$, we can apply \cref{claim:both_writers_can_lose_starting_in_safe} to $C_{i-1}$, where $r_i$ is the fixed register.
For $i \in \{1, ..., k\}$, let $\sigma'$ and $P_i$ be the schedule and set of processes achieved by applying \cref{claim:both_writers_can_lose_starting_in_safe} to configuration $C_{i-1}$ and the fixed register $r_i$.
Then let $\sigma_i = \Big(\Sched{C_{i-1}} | \big(P_i \cup \Lost(C_{i-1})\big)\Big) \circ \sigma'$.

By \cref{claim:both_writers_can_lose_starting_in_safe} Part~(a),~(b), and~(d), the inductive hypothesis is true.

Let $\sigma = \sigma_{k+1}$, and $Q = P_k \setminus \Lost(C)$.
(IH1) implies Part~(a).
Since at most $|P|/(10\ell)$ registers are poised to being accessed in $C$, we have $k \leq |P|/(10\ell)$.
Further, by \cref{claim:both_writers_can_lose_starting_in_safe} Part~(e), for each $i \in \{1, ..., k\}$ at most 2 processes are in $\Lost(C') \setminus \Lost(C)$.
Hence, by the inductive hypothesis,
\begin{equation}
|Q \setminus \Lost(C')| \geq |P| - k(8\ell - 1) - 2k \stackrel{\ell \geq 1}{\geq} |P| - 9k\ell \stackrel{k \leq |P|/(10\ell)}{\geq} |P| - \frac{|P|}{10\ell} 9\ell \geq |P| - \frac{9}{10} |P| \geq \frac{|P|}{10}.
\end{equation}
This proves Part~(b).
From Part~(c) of \cref{claim:both_writers_can_lose_starting_in_safe}, Part~(c) follows.
Since any process that receives the abort signal in $\E{C}$ is in $\Lost(C)$, and by Part~(d) of \cref{claim:both_writers_can_lose_starting_in_safe}, Part~(d) follows.
\end{proof}

\subsubsection{Detailed Construction}
\label{subsec:construction}
Let $n \geq 4$, $c = 10$, and $\ell = \lfloor \log n / (c \log\log n) \rfloor$.
We inductively construct a schedule $\sigma_i$ and a set of processes $P_i \subseteq \mathcal{P}$, for all $i \in \{0, ..., \ell\}$.
For the sake of conciseness, let $E_i = \Exec(\Gamma, \sigma_i)$, $C_i = \Conf(\Gamma, \sigma_i)$, and $L_i = \Lost(C_i)$.

The following invariants are satisfied for $i \in \{0, ..., \ell\}$:
\begin{enumerate}[(\text{I}1)]
  \item $C_i$ is safe.
  \item $|P_i \setminus L_i| \geq (n-1) / (\log n)^{ci}$.
  \item $\RMR_{P_i \setminus L_i}(E_i) \geq i\,|P_i \setminus L_i| - i$.
  \item For each process $p \in P_i \setminus L_i: \RMR_p(E_i) \leq i$.
  \item For each process $p \in P_i \setminus L_i$, $p^\top$ does not appear in $\sigma_i$.
\end{enumerate}
We now describe our inductive construction in detail.

\paragraph*{Base Case:}
Schedule $\sigma_0$ is a schedule in which each process scans its own shared memory segment, and $P_0 = \mathcal{P}$.
Note that $\Proc(\sigma_0) = \mathcal{P}$.
\paragraph*{Inductive Step:}
In $C_i$, we let each process in $P_i \setminus L_i$ that does not win in a solo-run take solo-steps until it is poised to perform an RMR.
By \cref{claim:zeroRMROneWins}, there is at most one process that wins in a solo-run starting in $C_i$, so in our execution all but one process participate.
By \cref{claim:projection_same_execution} each process performs the same steps in the solo-run starting in $C_i$ as in the constructed execution, and by \cref{claim:terminateorRMR} each process will eventually become poised to perform an RMR.
If there is a process that wins in a solo-run starting from $C_i$, we remove that process from the entire execution constructed so far.

More precisely, let $\{q_1, ..., q_k\} = P_i \setminus L_i$ and let $t_j$ be the largest integer, such that $\RMR\big(\Exec(C_i, q_j^{t_j})\big) = 0$ and $q_j$ does not terminate in $\Exec(C_i, q_j^{t_j})$, for $j \in \{1, ..., k\}$ (since by (I1), $C_i$ is safe, and by (I5) $q_j$ does not receive the abort signal in $\E{C_i}$, such an integer $t_j$ exists according to \cref{claim:terminateorRMR}).
By (I5), no process in $P_i \setminus L_i$ receives the abort signal in $\E{C_i}$. Thus, by \cref{claim:zeroRMROneWins}, any process that starting in $C_i$ terminates in its solo-run wins.
Hence, by the safety property of leader election, at most one process terminates in its solo-run starting in $C_i$.
If such a process does not exist, then let $\lambda_i = q_1^{t_1}q_2^{t_2}...q_k^{t_k}$, and $P_i' = P_i$. If such a process exists, we assume, without loss of generality (by renaming variables $q_1, ..., q_k$), that $q_k$ is the process that wins in its solo-run starting in $C_i$ without incurring any RMRs. Then let $\lambda_i = q_1^{t_1}q_2^{t_2}...q_{k-1}^{t_{k-1}}$, $P_i' = P_i \setminus \{q_k\}$.
Finally, let $D_i = \Conf\big(\Conf(\Gamma, \sigma_i | P_i'), \lambda_i\big)$.

We define for register $r$, set $R_i(r)$ as the set of processes that are poised to read $r$ in $D_i$, and set $W_i(r)$ as the set of processes that are poised to write $r$ in $D_i$.
Let $S_i=\{r \in \RR \,|\, W_i(r) \cup R_i(r) \neq \emptyset\}$.

First, we prove some properties of configuration $D_i$. 
\begin{claim}
  \label{claim:DiProperties}
  The following are true for configuration $D_i$.
  \begin{enumerate}[(a)]
    \item Configuration $D_i$ is safe.
    \item For any process $p \in P_i' \setminus L_i$, $\RMR_p(\E{C_i}) = \RMR_p(\E{D_i})$.
    \item For any process $p$ in $R_i(r) \cup W_i(r)$, it is true that $r \notin \RR_p$.
  \end{enumerate}
\end{claim}
\begin{proof}
  Let $k' = |P_i' \setminus L_i|$.
  Since $C_i$ is safe, by \cref{claim:projection_same_execution},
  \begin{equation}
  \label{eq:PiPrimeExecution}
  \Exec(\Gamma, \sigma_i | P_i') = \E{C_i} | P_i'.
  \end{equation}
  Further, by \cref{claim:projection_cache}, for each process $p \in P_i'$, it holds that $\Cache_p(C_i) = \Cache_p\big(\Conf(\Gamma, \sigma_i | P_i')\big)$.
  Therefore, when $\RMR\big(\Exec(C_i, q_j^{t_j})\big) = 0$, for $j \in \{1, ..., k'\}$, it is true that $\RMR\Big(\Exec\big(\Conf(\Gamma, \sigma_i | P_i'), q_j^{t_j} \big)\Big) = 0$.
  Thus, by applying \cref{claim:zeroInformation} Part~(b) several times (first $Q_1 = \{q_1\}$ and $Q_2 = \{q_2\}$, the next time, $Q_1 = \{q_1, q_2\}$ and $Q_2 = \{q_3\}$, and so on),
  \begin{equation}
  \label{eq:LambdaZeroRMR}
  \RMR\Big(\Exec\big(\Conf(\Gamma, \sigma_i | P_i'), \lambda_i \big)\Big) = 0.
  \end{equation}
  Hence, since by (I5) none of the processes in $P_i' \setminus L_i$ receive the abort signal in $\Conf(\Gamma, \sigma_i | P_i')$ or during $\Exec\big(\Conf(\Gamma, \sigma_i | P_i'), q_j^{t_j} \big)$, by applying \cref{claim:zeroRMRsafe} multiple times, $D_i$ is safe.
  This proves Part~(a).
  
  By \cref{eq:PiPrimeExecution}, for each process $p \in P_i'$, it holds that $\RMR_p(\E{C_i}) = \RMR_p\big(\Exec(\Gamma, \sigma_i | P_i')\big)$.
  Thus, Part~(b) follows from \cref{eq:LambdaZeroRMR}.
  
  By \cref{claim:winnerProjection}, each process in $P_i' \setminus L_i$ has the same cache in $C_i$ and $\Conf(\Gamma, \sigma_i | P_i')$.
  Hence, by \cref{eq:PiPrimeExecution} and the construction of $\lambda_i$, each process in $R_i(r) \cup W_i(r)$ is poised to perform an RMR step in $D_i$.
  Therefore, the register that each process $p \in R_i(r) \cup W_i(r)$ is poised to access is not in its own memory segment.
\end{proof}

Let $X_i = \bigcup_{r \in S_i} W_i(r)$, and $Y_i = \bigcup_{r \in S_i} R_i(r)$.
We distinguish the following cases to complete the inductive step of our construction:
\paragraph*{Case 1: $|S_i| \geq |P_i \setminus L_i| / (10 \ell)$ or $|X_i| < |Y_i|$:}
Let $\sigma_{i+1} = \sigma$ and $P_{i+1} = Q \cup \Lost(D_i)$, where $Q$ and $\sigma$ are the set of processes and the schedule we know from Lemma \ref{lemma:LowContentionAndRead}.
From Part~(b), we get that $C_{i+1}$ is safe.
By Part~(a),
\begin{equation}
|P_{i+1} \setminus L_{i+1}| \geq \frac{|P_i|}{60\ell^2}
\stackrel{(I2)}{\geq} \frac{n-1}{(\log n)^{ci}60\ell^2}
\stackrel{\ell \leq \log n}{\geq} \frac{n-1}{60(\log n)^{ci+2}} \stackrel{c = 10, n \geq 4}{\geq} \frac{n-1}{(\log n)^{c(i+1)}}.
\end{equation}
Hence, (I2) is true.
From (I3), (I4), and Part~(c), it immediately follows that (I3) and (I4) are true for $i+1$.
Invariant (I5) directly follows Part~(d).

\paragraph*{Case 2: $|S_i| < |P_i \setminus L_i| / (10 \ell)$ and $|X_i| \geq |Y_i|$:}
Applying Lemma \ref{lemma:high_contention_write_case} to configuration $D_i$, results in a set of processes $Q$, and a schedule $\sigma$.
Let $P_{i+1} = Q \cup \Lost(D_i)$ and $\sigma_{i+1} = \sigma$.
From Lemma~\ref{lemma:high_contention_write_case} we prove Invariants~(I1)-(I5).
Invariant~(I1) follows Part~(a).
From Part~(b), we have
\begin{equation}
|P_{i+1} \setminus L_{i+1}| \geq |P_i \setminus L_i|/10 \stackrel{(I2)}{\geq}
\frac{n-1}{10(\log n)^{ci}} \stackrel{n \geq 4, c = 10}{\geq} \frac{n-1}{(\log n)^{c(i+1)}},
\end{equation}
which proves (I2).
From (I3), (I4), and Part~(c), Invariants (I3) and (I4) are true.
Part~(d) and (I5) immediately imply (I5).

\paragraph*{Proof of \cref{thm:main_lower_bound}}
Using Invariants (I1)-(I5), we obtain our main theorem.
As shown above, for any abortable leader election algorithm, there exists an execution $\Exec(\Gamma, \sigma_{\ell - 1})$ that satisfies (I1)-(I5).
We have
\begin{equation}
  \frac{n-1}{(\log n)^{c(\ell - 1)}} \geq \frac{n-1}{(\log n)^{(\log n / \log \log n) - c}} \geq \frac{(n-1)(\log n)^c}{n} \stackrel{n \geq 4}{\geq} 2.
\end{equation}
Hence, by (I2) in $\Exec(\Gamma, \sigma_{\ell - 1})$ at least two processes participate and don't lose.
By (I3) at least one of these processes incurs $\Omega(\ell)=\Omega(\log n / \log\log n)$ RMRs.

\clearpage
\bibliography{main}
\clearpage
\end{document}